\documentclass{article}
\usepackage{amsmath, amsfonts,amssymb, graphicx, setspace, subcaption}
\usepackage{enumitem}
\usepackage{authblk}
\usepackage[
backend=biber,
style=alphabetic,
sorting=nyt
]{biblatex} 
\usepackage[export]{adjustbox}
\providecommand{\keywords}[1]
{
  \small	
  \textbf{\textit{Keywords---}} #1
}

\newtheorem{theorem}{Theorem}
\newtheorem{lemma}{Lemma}

\newcommand{\qed}{\hfill\blacksquare}
\setlist[itemize]{after={\bigskip}}

\begin{document}
\title{The information-theoretic complexity of differentiable functions}
\author[]{Matthijs Ruijgrok}
\affil[]{Mathematics Department, Utrecht University}
\affil[]{\textit {m.ruijgrok@uu.nl}}
\date{\today}

\maketitle
\begin{abstract}
    \noindent A measure for the complexity of a differentiable function on an interval is introduced. It is based on approximations of the function by piecewise constant functions. The measure takes into account the quality of the approximation and the number of intervals in the approximating function. \\
    This measure, called the V-complexity of $f$, is shown to formalize some intuitions about the simplicity or complexity of $f$.\\
    The V-complexity is then compared to another measure of complexity, namely how compressible an approximation of $f$ is. It is hypothesized that V-complexity is equivalent to the compression measure, in the case of the Run Length Encoding and the Lempel Ziv 77 algorithms.\\
    V-complexity can be used as an ingredient in the definition of the Effective Complexity (EC) of a Complex System. When the perceived regularities of such a system are described by a differentiable function on an interval, the EC can be defined as the V-complexity of that function. EC is applied to the model of diffusion of cream in a cup of coffee. The perceived regularity of this model is given by the diffusion equation. The V-complexity of the solution of the equation starts at zero, quickly increases to a maximum and then decreases back to zero as the liquid reaches its equilibrium state. It is shown that this is also the result when a cellular automaton approach and the concept of Apparent Complexity is used.     
\end{abstract}
\section{Introduction}

The books, papers and internet tutorials on complexity agree on one thing: it is a concept that is hard to pin down, see \cite{Ladyman} and \cite{Newman}. One reason is that complexity is used to describe systems of various intricacy, from the global climate down to one-dimensional cellular automata. This paper stays at the lower end of the scale and proposes a measure for the complexity of differentiable functions $f: [a,b] \rightarrow \mathbb{R}$. The notion of complexity that will be used is the amount of information needed to describe an object \cite{Prokopenko}. \\
Figure 1 gives a sample of functions with a variety of features. Most likely, the reader will assign the highest complexity to functions with a lot of detail. Formalizing this intuition is the goal of this paper.
\begin{figure}[h!]
\centering
\begin{subfigure}{0.3\textwidth}
\centering
\includegraphics[scale=0.2]{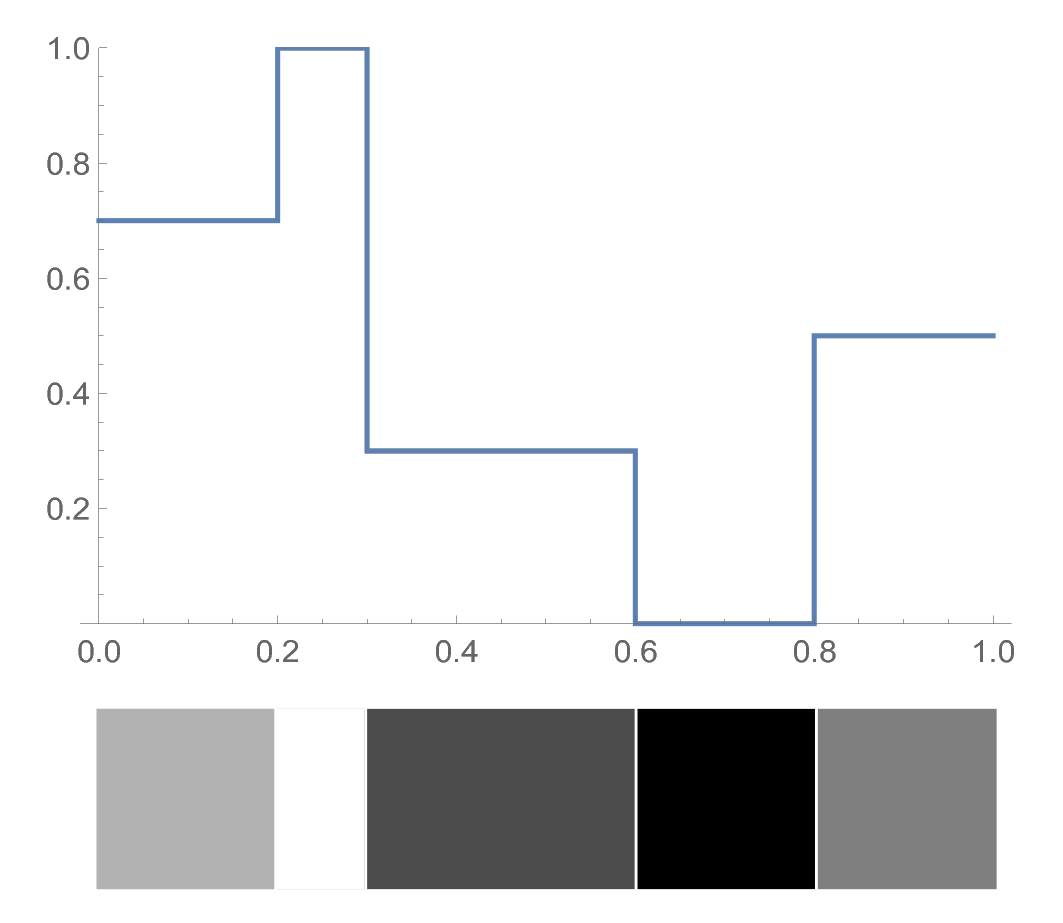}
\caption{$V(f)=0$}
\label{d0r45}
\end{subfigure}
\begin{subfigure}{0.3\textwidth}
\centering
\includegraphics[scale=0.2]{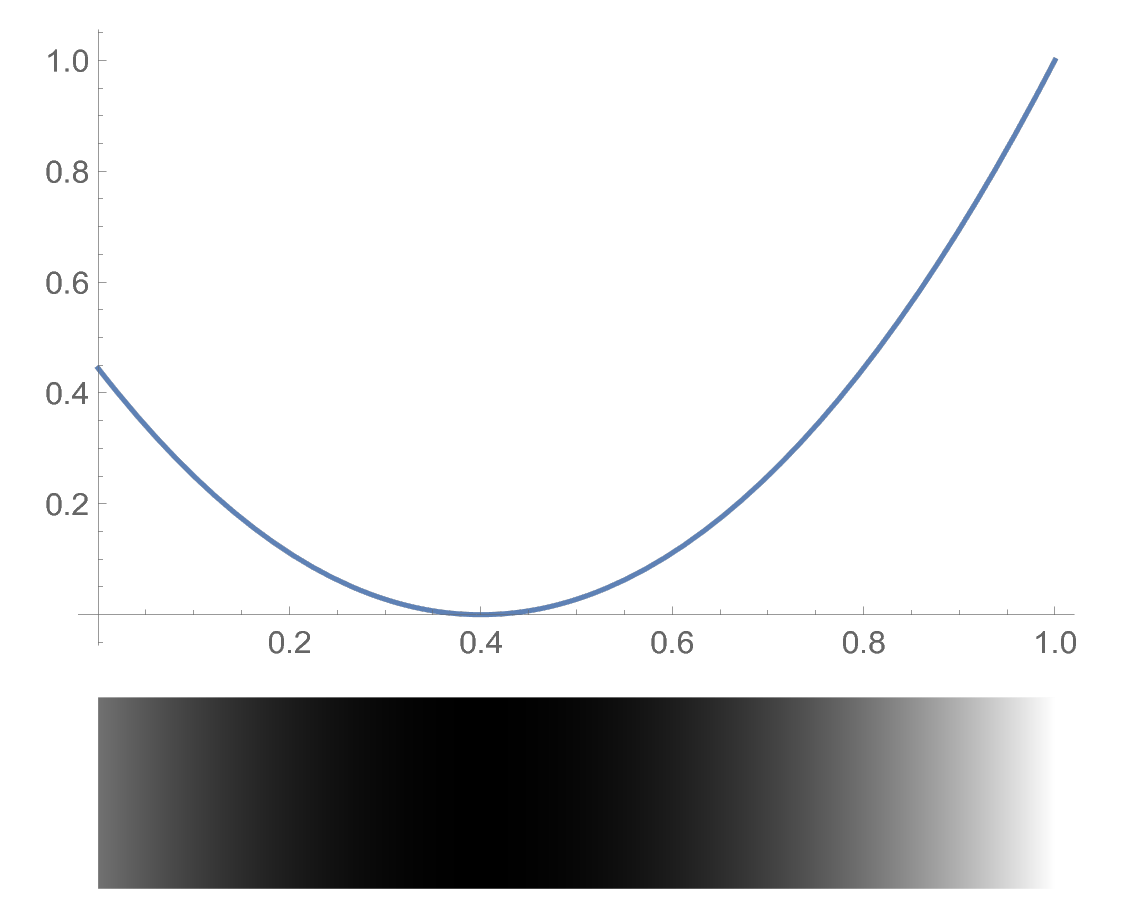}
\caption{$V(f)=1.27$}
\label{d1r45}
\end{subfigure}
\begin{subfigure}{0.3\textwidth}
\centering
\includegraphics[scale=0.2]{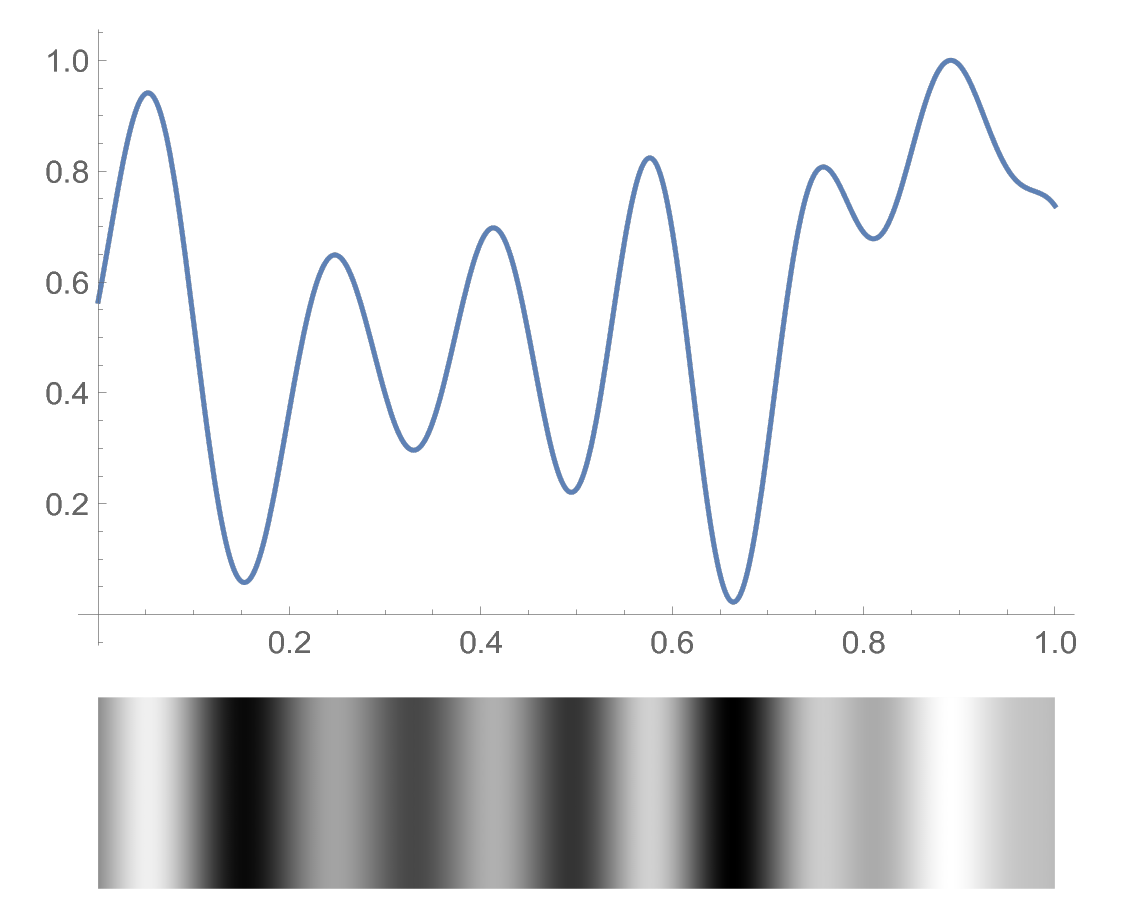}
\caption{$V(f)=5.24$}
\label{d2r45}
\end{subfigure}
\caption{Some functions with domain and range equal to $[0,1]$, and their complexities according to expression (\ref{Sexp}). Below the graphs are the corresponding grayscale figures.}
\label{rule45}
\end{figure}
\newline
As the amount of information contained in a function $f: [a,b] \rightarrow \mathbb{R}$ is generally infinitely large, the complexity of $f$ will be based on finite approximations of the function. In a nutshell, the easier it is to describe a good approximation of a function, the lower its complexity.\\
To make this idea more concrete, some choices have to be made. First of all, a good approximation of a function $f$ is another function $g$ such that the norm of the difference $||f-g||$ is small. There is a variety of norms on function spaces, but the most commonly used ones are the $L_p$ norms.  The choice in this paper is to use the $L_1$ norm. This choice will be motivated in section \ref{sectappr}.\\
Second, the ease of description of an approximation of a function $f$ needs to be addressed. There exist sets of functions such that every differentiable function can be written as an infinite series expansion in terms of elements from this set.  The approximation of $f$ will consist of a finite number of terms of the expansion. Given such a set, $f$ is said to be easily described if a small number of terms of the expansion of $f$ already give a good approximation. One example is the Fourier set of sine and cosine functions. Every differentiable function on an interval can be written as an infinite expansion in terms of these functions.\\
The choice of the set of function that are to be used in the expansion is therefore essential. In this paper, the step functions will be used. These are functions that are piecewise constant. It is known that all continuous functions on $[a,b]$, which includes the differentiable functions, can be approximated in the $L_1$ norm by a convergent series of step functions \cite{Royden}.\\
The motivation for choosing the step functions is that it connects to properties of our visual perception. Humans, and many other species, are very good at distinguishing contrasts and outlines. To the human eye, a graph such as in Figure 1a is simple. It can be described by a few real numbers which identify the intervals where the function is constant, and the values of the function on these intervals. Functions which we perceive as less simple, such as Figure 1c, have many wiggles and details on a small scale. A good approximation in the form of a step function needs many intervals.\\
The choice would be different if it was based on human audio perception. There, the Fourier basis would be the appropriate choice, since the ear works as a natural Fourier transformer for many types of signals \cite{Lewicki}. To the ear, Figure 1c, which consists of fifteen Fourier modes, is the relatively easy function, whereas a step function signal needs many more modes to describe adequately. In the Fourier world, a step function is highly complex.\\\\
In section \ref{sectappr}, the main result of this paper is formulated. Given a differentiable function $f:[a,b]\rightarrow {\mathbb{R}}$ and a maximal error $\varepsilon>0$, there exists a step function $q$ such that $||f-q||_1<\varepsilon$ and $q$ uses a minimal number of intervals $N_{\varepsilon}(f)$.  As $\varepsilon \rightarrow 0$
\begin{align}
    \label{firstS}
    N_{\varepsilon}(f)=V(f)\,\varepsilon^{-1} \, .
\end{align}
In other words, the number of intervals needed for a good approximation of $f$ is of order ${\cal{O}}(\varepsilon^{-1})$ as $\varepsilon \rightarrow 0$, for all differentiable functions $f$. The pre-factor $V(f)$ is the definition of the complexity of $f$ and it will be referred to as the V-complexity of $f$.\\
The accuracy of the approximation increases as $\varepsilon$ decreases. Therefore, the quantity $\varepsilon^{-1}$ can be seen as a measure of the accuracy of the approximation. Rewriting (\ref{NQf}) as
\begin{align}
\label{QNe}
    V(f)=N(\varepsilon)/\varepsilon^{-1}=N(\varepsilon)\varepsilon 
\end{align}
shows that the V-complexity of $f$ is the ratio of the minimal amount of information needed to describe $f$ to the accuracy of the description, as the accuracy becomes perfect. Alternatively, V-complexity is the product of this minimal amount of information with the error of the approximation, as the error goes to zero.\\
The proof of the result consists of constructing an approximating step function $q$ such that the error on every interval is the same and the total error is $\varepsilon$. This is known as the equidistribution of error approximation and there is good evidence that this approximation uses the minimal number of intervals among all step function approximations. The result is that
\begin{align}
\label{Sexp}
    V(f)=\frac{1}{4}\left( \int_a^b |f'(x)|^{\frac{1}{2}} \, \textrm{d}x\right)^2\,.
\end{align}
In section \ref{exmp}, a number of examples will show that V-complexity is small for functions which themselves are close to a step function. For functions with many wiggles, such as in Figure \ref{d2r45}, the V-complexity is relatively large. \\
\\
Before continuing with the main application of V-complexity, an alternative complexity measure for differentiable functions, named compression-complexity, is introduced in section \ref{Qandcomp}. The idea is to discretize the function $f$ by constructing a regular grid on the domain and on the range, with grid widths $\Delta x$ and $\Delta y$, respectively. The function is sampled on the grid  points $x_k=a+k\Delta x$ and the outcomes rounded to the nearest multiple of $\Delta y.$. This procedure transforms $f$ to a finite string of symbols from a finite alphabet. Then, a lossless compression algorithm ${\cal{A}}$ is applied to this string. Analogous to definition (\ref{QNe}), the ${\cal{A}}$-complexity of $f$ is defined as length of the compressed string multiplied by the error of the approximation, as the error goes to zero.\\
The compression algorithms considered are the simple Run Length Encoding and the more sophisticated Lempel-Ziv 77 encoding, used in the GZIP tool and many others. The result of this section is the hypothesis that V-complexity is close to RLE-complexity and that RLE and LZ77 complexity are equivalent, as $\Delta x$ goes to zero.\\ \\
V-complexity is an information-theoretic concept, applicable to functions, which are static and deterministic objects. However, there is a link with Complex Systems, understood as dynamical systems consisting of many interacting elements and with possibly a stochastic component. This link is investigated in Section \ref{CCex}. The connection is through Effective Complexity, which is defined as "the length of a concise description of the system's perceived regularities" \cite{GellMann}. It is based on the idea that a Complex System can be decomposed into a regular part and a residual part which is irregular or probabilistic.\\
This definition, however, assumes that the regular part of the system is given by a finite set of numbers. There are situations where the regular part is best described by a function $f$ on an interval of real numbers. In those cases, even the most concise description of the perceived regularities has infinite length. It is here that V-complexity is useful, by changing the definition of Effective Complexity slightly to "the V-complexity of the function $f$ which describes a system's perceived regularities".\\
 This modification of Effective Complexity is applied to the paradigmatic example of diffusion of cream in a cup of coffee. In the model, a cup of cappuccino starts out in an initial state where the upper half consists of only cream and the lower half only of coffee. The particles start to mix and this diffusion process eventually leads to a uniformly mixed liquid.  The model was introduced in \cite{Carroll} to illustrate the difference between the concepts of entropy and of complexity. It will be studied first through computer simulations of a cellular automaton and compression complexity, an approach called Apparent Complexity. Second, the model is studied analytically, using the diffusion equation and the modified notion of Effective Complexity.\\
The outcome is that both approaches lead to qualitatively similar results, namely that complexity first increases, reaches a maximum and then decreases to zero, as the equilibrium state is approached.\\
Finally, section \ref{Concl} concludes.

\section{Approximation by step functions}
\label{sectappr}
To measure the nearness of an approximation, a norm has to be chosen from the family of $L_p$ norms. For $p \geq 1$ and a function $f:[a,b]\rightarrow {\mathbb{R}}$, this norm is defined as:
\begin{align*}
    ||f||_p=\left(\int_a^b |f(x)|^{p}\, \textrm{d}x\right)^{1/p} \, ,
\end{align*} 
These norms are not equivalent. On the one hand, Hölder's inequality states that for all $p\geq 1$ 
\begin{align*}
    ||f||_1 \leq C_p ||f||_p \, , \quad C_p=(b-a)^{1-1/p} \, .
\end{align*}
This inequality implies that if a function $g_{\varepsilon}$ is an ${\cal{O}}(\varepsilon)$ approximation of a function $f$ in an $L_p$ norm for any $p >1$, it is also an ${\cal{O}}(\varepsilon)$ approximation in the $L_1$ norm.\\
However, the converse is not true. As an example, let $f$ be the step function that is equal to 0 if $x<0$ and equal to 1 for $x\geq 0$. As a grayscale picture, $f$ is a Malevich style painting, with the left half completely black and the right half completely white. The boundary between black and white is at $x=0$. Take $g_{\varepsilon}$ to be the same picture, but with the boundary at $x=\varepsilon$. Then, $||f-g_{\varepsilon}||_p={\varepsilon}^{1/p}$. Therefore, the function $g_{\varepsilon}$  is an ${\cal{O}}(\varepsilon)$ approximation of $f$ only for $p=1$. For $p>1$ the distance between $f$ and $g_{\varepsilon}$ is asymptotically larger.\\
Continuing the human vision analogy, it is desirable that in the above example the difference between $f$ and $g_{\varepsilon}$ should be small. The contrast between the black and the white halves of the picture stands out, not the precise location of the boundary. Therefore, the choice of norms should make the difference minimal, which corresponds to $p=1$.\\\\  
For a given $\varepsilon>0$, a function $f(x)$ with domain $[a,b]$ will be approximated by a step function $q_{\varepsilon}(x)$ in such a way that $||f-q_{\varepsilon}||_1\leq \varepsilon$. This approximation should be achieved by using the smallest number of jumps in $q_{\varepsilon}(x)$. Finding such a $q_{\varepsilon}(x)$ is a very difficult optimization problem. Instead, $q_{\varepsilon}(x)$ will be restricted so that the errors are equidistributed, meaning that all the errors on each subinterval where $q_{\varepsilon}(x)$ is constant are the same. In \cite{Huang} it is shown that equidistribution of errors leads to optimal approximations as $\varepsilon \rightarrow 0$, under certain conditions. A full analysis on whether these conditions are exactly met in this case would lead too far, but the least that can be said is that equidistribution approximations perform very well in practice \cite{DamZegeling}.\\
Let the total number of steps in $q_{\varepsilon}(x)$ be $N(\varepsilon)$. Then $q_{\varepsilon}(x)$ is defined by the set $\{\{I_1,q_1\}, \ldots,\{I_N,q_{N(\varepsilon)}\}\}$. Here, the $I_k$ are disjoint intervals and $q_{\varepsilon}(x)=q_k$ for all $x \in I_k$. Equidistribution of errors means that
\begin{align}||f-q_k||_{I_k}&= \varepsilon / N(\varepsilon) \, , \quad k=1, 2, \ldots, N(\varepsilon) \label{equi1}\\
 \bigcup_{k=1}^{N(\varepsilon)} I_k&=[a,b] \label{equi2}\, ,
\end{align}
where $||f-q_k||_{I_k}$ is shorthand for the $L_1$ norm of the difference $f-q_k$, restricted to the interval $I_k$.\\
The aim of the following algorithm is to minimize the number $N(\varepsilon)$ of intervals needed for the approximation, assuming that the total error is equidistributed. This is achieved by taking $I_1=[a,x_1]$ and calculating the largest value of $x_1$ such that for an optimal $q_1$, condition (\ref{equi1}) is satisfied. Taking $I_2=[x_1,x_2]$ this procedure is repeated to calculate $x_2$. The algorithm ends when an index $k$ is reached such that $x_k\geq b$. This index is the desired $N(\varepsilon)$.\\ 
An asymptotic expansion for $N(\varepsilon)$, as $\varepsilon \rightarrow 0$ is found by transforming the recursion used to calculate $x_{k+1}$ from $x_k$ into a differential equation. This results in:
\begin{theorem} 
\label{thm1}
Let $f \in C^1([a,b])$. Then, asymptotically as $\varepsilon \rightarrow 0$,
\begin{itemize} 
    \item [(a)] 
    \begin{align}
    \label{NQf}
        N(\varepsilon)=V(f)\,\varepsilon^{-1}\, ,
    \end{align}
    with
    \begin{align}
    \label{Sp}
        V(f)=\frac{1}{4}\left(\int_a^b|f'(x)|^\frac{1}{2}\textrm{d}x\right)^{2}\,.
    \end{align}
    \item [(b)] Let $y(s)$ be the solution of the differential equation 
    \begin{align}
    \label{diffy}
       |f'(y)|^{\frac{1}{2}} \frac{\textrm{d}y}{\textrm{d}s}=\left( 4 V(f)\right)^{\frac{1}{2}} \, , \quad y(0)=a.
    \end{align} 
Then, asymptotically as $\varepsilon \rightarrow 0$,
\begin{align}
\label{xk}
    x_k=y(\frac{\varepsilon k}{V(f)}) \, , \quad k=0,1,\ldots, N(\varepsilon)\, .
\end{align}
and
\begin{align}
    q_k=f((x_{k-1}+x_k)/2) \, , \quad k=1,2,\ldots, N(\varepsilon)\, .
\end{align}
\end{itemize}
\end{theorem}
The proof of this theorem can be found in Appendix A.\\ \\
Part ($a$) of the theorem states, given a maximal approximation error $\varepsilon$, the number of intervals needed for an approximation scales as $\varepsilon^{-1}$. The pre-factor $V(f)$ is taken as the definition of the complexity of $f$.\\
Part ($b$) of the theorem gives a method for constructing the approximating step function. In general, the lengths of the intervals $I_k$ are not equal. The domain $[a,b]$ is partitioned into a non-uniform grid. In the language of Numerical Analysis, equation (\ref{diffy}) shows that this partition uses a version of the Beckett-Mackenzie monitor function \cite{Qiu}. From (\ref{diffy}) it follows that the intervals are shorter where $|f'(x)|$ is large and vice-versa.

\section{Examples}
\label{exmp}
\begin{figure}[h]
\centering
\includegraphics[width=0.6\textwidth]{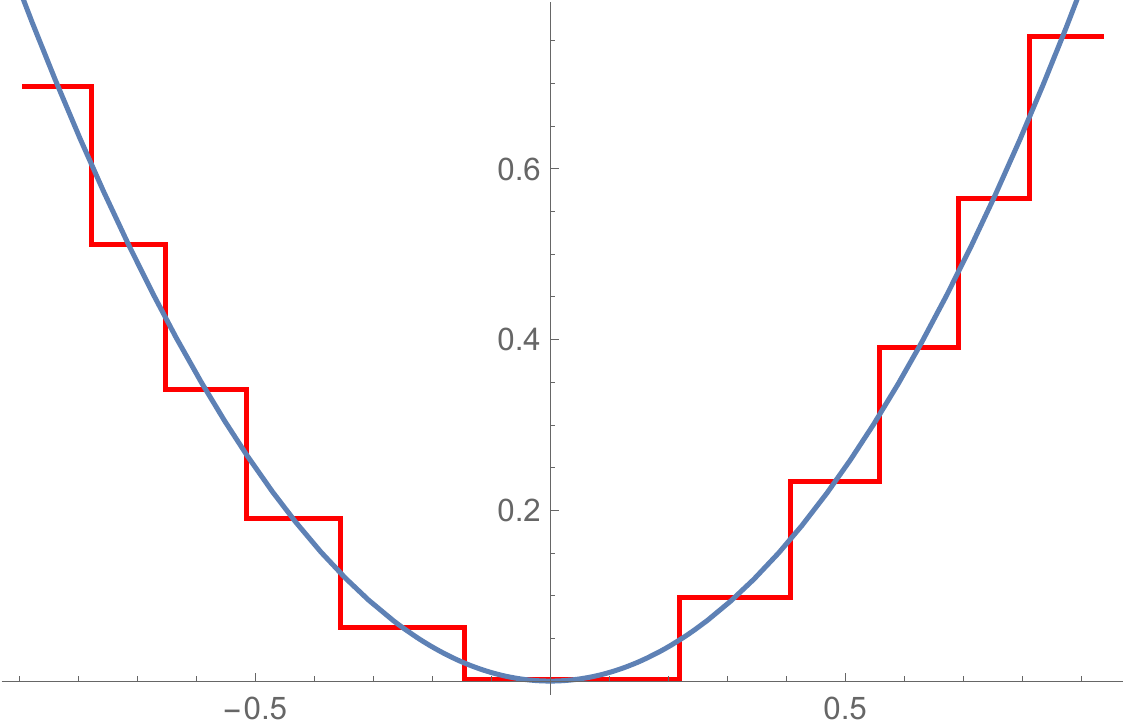}
\caption{The step function approximation for $f(x)=x^2$, with $\varepsilon=0.07$.}
\label{parabola}
\end{figure}
\noindent 
In the following list, $V(f)$ is calculated for some (families of) functions which demonstrate essential properties of this complexity measure. 
\begin{itemize} 
\item[(a)] {\bf{Quadratic function}} $f(x)=x^2$ on $[-1,1]$.  From (\ref{Sp}) it follows that
\begin{align*}
    V(f)=\frac{1}{4}\left( \int_{-1}^{1}|2x|^{\frac{1}{2}} \textrm{d}x \right)^2=\frac{8}{9} \, .
\end{align*}
The differential equation (\ref{diff}) becomes:
\begin{align*}
    \frac{\textrm{d}y}{\textrm{d}s}=\frac{4}{3} |y|^{-\frac{1}{2}}\, , \quad y(0)=-1 \, ,
    \end{align*} 
    with solution
    \begin{align*}
        y(s)=\begin{cases} -(1-2s)^{2/3} & \text{if } 0 \leq s \leq 1/2 \\
                      (2s-1)^{2/3} & \text{if } 1/2 <s \leq 1 \, .
        \end{cases}
    \end{align*}
Hence, asymptotically as $\varepsilon \rightarrow 0$,
\begin{align*}
x_k=\begin{cases} -(1-9 \varepsilon k/4)^{2/3} & \text{if } 0 \leq k \leq 4/(9 \varepsilon) \\
                      (9 \varepsilon k/4-1)^{2/3} & \text{if } 4/(9 \varepsilon) <k \leq 8/(9 \varepsilon) \, .
        \end{cases}
\end{align*}

\item [(b)] {\bf{Step functions}}. Let $[a,b]=[-1,1]$ and  $\alpha>0$. Define 
\begin{align*}
    f_{\alpha}(x)=\frac{1}{1+e^{-\alpha x}} \, .
\end{align*}
This is a sigmoid function which converges in the $L_1$ norm to the step function $\theta(x)$ as $\alpha \rightarrow \infty$. Here, $\theta(x)$ is defined as $\theta(x)=0$ when $x<0$, $\theta(0)=1/2$ and $\theta(x)=1$ when $x>0$. This follows from 
\begin{align*}
    \int_{-1}^1|f_\alpha(x)-\theta(x)|\,\text{d}x=1-\frac{1}{\alpha}\log((2+e^\alpha+e^{-\alpha})/4)\, .
\end{align*}
The V-complexity of $f_\alpha$ is:
\begin{align*}
V(f_{\alpha})&=\frac{\alpha}{4}(\int_{-1}^{1} (\frac{e^{-\alpha x}}{(1+e^{-\alpha x})^2})^{1/2}\textrm{d}x)^2 \\
&=\frac{1}{\alpha}\left( \arctan e^{\alpha/2}-\arctan e^{-\alpha/2}\right)^2 \, .
\end{align*} 
Therefore, 
\begin{align*}
\lim_{\alpha \rightarrow \infty}Q(f_{\alpha})=0 \, .   
\end{align*}
It is shown in Appendix B that for every sequence $\{f_n(x)\}$ of monotonically increasing differentiable functions, which converges in $L_1$ norm to the step function $\theta(x)$, the complexity $V(f_n)$ converges to zero.\\
The condition that the $f_n(x)$ be increasing functions is important. If this is not the case, the statement is likely not true in generality. This situation needs further study.
\item [(c)] {\bf{Power functions}}. Let $[a,b]=[0,1]$ and $f_{\alpha}(x)=x^\alpha$, with $\alpha \geq 0$. A straightforward integration yields
\begin{align*}
    V(f_{\alpha})=\frac{\alpha}{(\alpha +1)^2}\, .
\end{align*} 
For $\alpha \rightarrow \infty$, c.q. $\alpha \rightarrow 0$ the function $f_{\alpha}(x)$ is flat, save for a boundary layer near $x=1$, c.q. $x=0$. This is reflected in the fact that for both these limits, $V(f_{\alpha}) \rightarrow 0$.\\
The maximal complexity occurs when $\alpha=1$, where $V(f_1)=1/4.$
\item [(d)] {\bf{Oscillating functions}} Let $f_n(x)=\cos(n \pi x) $, with  $x \in [0,1]$ and $n$ an integer. For large $n$, this is a rapidly oscillating function. Then, since $|\sin(n \pi x)|$ is periodic with period $1/n$, 
\begin{align*}
    V(f_n)&= \frac{1}{4}\left(\int_0^1 |n \pi\sin(n \pi x)|^{\frac{1}{2}}\, \textrm{d}x\right)^{2}\\
    &=\frac{1}{4} n^{3} \pi\left( \int_0^{1/n}| \sin(n \pi x)|^{\frac{1}{2}}\, \textrm{d}x\right)^{2} \, .
    \end{align*}
 Substituting $y=n \pi x$ gives 
 \begin{align*}
 V(f_n)=\frac{n}{4\pi} \left(\int_0^\pi (\sin(y)^{\frac{1}{2}} \textrm{d}y\right)^{2}\approx0.457\, n \, . 
 \end{align*}
The complexity grows linearly with the number of oscillations. This example shows that the complexity of a function can be arbitrarily large, even for bounded functions.
     
\end{itemize}

\section{V-complexity and compression}
\label{Qandcomp}
An alternative method to define the minimal information necessary to describe a function is to use a compression algorithm. An approximation of the function is encoded in a finite string of symbols from a finite alphabet. The string is then compressed, using some lossless algorithm ${\cal{A}}$. The size of the compressed string is taken as the measure of the amount of information needed to describe the function. The approximation has a certain error $\varepsilon$, and the ${\cal{A}}$ complexity  is defined analogous to (\ref{QNe}).\\
The encoding of $f:[a,b]\rightarrow {\mathbb{R}}$ starts with dividing both the domain and the range into regular grids, with grid sizes $\Delta x$ and $\Delta y$, respectively. Equivalent parameters are $\Delta x$ and $r=\Delta y/\Delta x$. The domain has $n=(b-a)/\Delta x$ sub intervals. The function is sampled at the midpoints $x_k$ of the sub intervals and the values $y_k=f(x_k)$ are rounded to the nearest multiple of $\Delta y$. \\ 
The discretization yields outcomes $y_k=a_k \Delta y$, $k=1,\ldots n$, with $a_k$ nonnegative integers. Concatenating these outcomes, and leaving out the $\Delta y$'s, gives the string ${\bf{s}}=a_1\ldots a_{n-1}$. This string is the encoding of the step function $q_{u}(x;\Delta x,r)$ which has the value $a_k \Delta y$ for all $x \in [a+k \Delta x,a+(k+1)\Delta x]$.
\begin{figure}[h]
\centering
\includegraphics[width=0.6\textwidth]{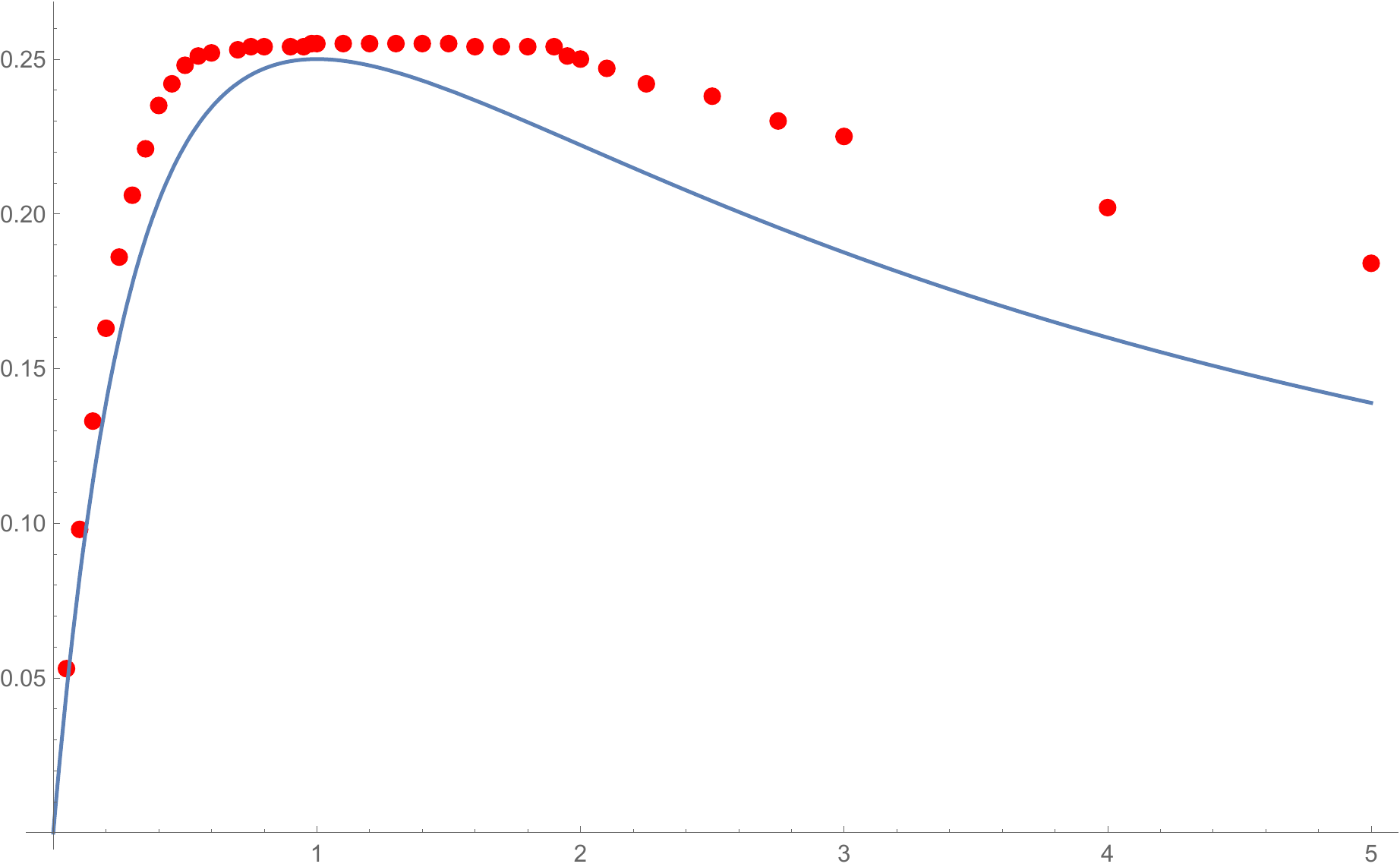}
\caption{Complexity of $f_{\alpha}(x)=x^\alpha$. On the horizontal axis are the values of $0\leq \alpha\leq 5$. Red dots indicate the RLE-complexity $C(f_\alpha)$, the blue graph shows $V(f_\alpha)$.}
\label{RQcomp}
\end{figure}
\subsection{RLE-complexity}
One compression algorithm that can be used is Run Lenth Encoding (RLE). This method counts successive repetitions of a symbol within a string and replaces it with the symbol and the number of repetitions. As an example, let $f(x)=x^4$, $[a,b]=[0,1]$, $\Delta x=0.1$ and $r=1$. Then ${\bf{s}}=0000012358$ and\\ 
$RLE({\bf{s}})=\{(0,5),(1,1),(2,1),(3,1),(5,1),(8,1)\}$. The number of pairs of $RLE({\bf{s}})$ will be denoted by $N_{RLE}(\Delta x,r)$ and is equal to $6$ in this example. The $L_1$ error $||f-q_{u}||_1$ will be denoted by $\varepsilon(\Delta x,r)$. In this example it is equal to $0.03$.\\
The RLE-complexity of a function $f$ is defined as:
\begin{align}
    C(f)=\min_{r \geq 0}\lim_{\Delta x \rightarrow 0}N_{RLE}(\Delta x,r)\,\varepsilon(\Delta x,r)\, ,
\end{align}
assuming the limit exists. As with V-complexity, the RLE-complexity is the product of the amount of information in the RLE-approximation with the error of the approximation, as the error goes to zero. However, these two ingredients depend on the grid ratio $r$, so the minimum over $r$ is taken.\\
An explicit expression for $C(f)$ is desirable, but a preliminary analysis has shown that this is not a simple task, and it will not be pursued here. However, extensive numerical simulations have shown that
\begin{align}
\label{etagam}
    N_{RLE}(\Delta x,r)&=\eta (r)(\Delta x/(b-a))^{-1}=\eta(r)n(\Delta x) \nonumber \\
    \varepsilon(\Delta x,r)&=\gamma(r) \Delta x \, , 
\end{align}
as $\Delta x$ goes to zero.\\
The term $\eta(r)$ is the RLE compression factor. The number of bins in the range is inversely proportional to $\Delta y$, so also to $r$ for fixed $\Delta x$. This implies that $\eta(r)$ is a decreasing function. For $r=0$, every $y$ value is in it's own bin, so no compression is possible. Therefore, $\eta(0)=1$. For very large values of $r$ there are only a few available vertical bins, so the set of outcomes is compressed to an extreme degree and $\eta(r)$ becomes very small.\\
The term $\gamma(r)$ is an increasing function. As the size of the vertical bins increases, and keeping $\Delta x$ fixed, the accuracy of the approximation decreases, which means that the $L_1$ error increases.\\
From (\ref{etagam}) it follows that
\begin{align}
\label{Cf}
    C(f)=\min_{r \geq 0}\, \eta(r)\gamma(r)/(b-a) \, .
\end{align}
There is a connection between RLE-complexity and V-complexity. Returning to the example with $f(x)=x^4$, it can be noted that the step function $q_{u}$ encoded by $0000012358$ has the same value on the first five sub intervals. Hence, it can also be described as a step function on a nonuniform grid by concatenating the first five sub intervals in to one larger interval. The compressed set $RLE({\bf{s}})$ is the encoding of $q_{u}$ on this nonuniform grid.\\
Therefore, RLE-complexity and V-complexity, which is based on the number of intervals in step function approximations with a nonuniform grid, can be expected to be close.\\
Figure  \ref{RQcomp}) shows the two complexities for the family of functions $f_\alpha(x)=x^{\alpha}$ on the interval $[0,1]$. For the RLE encoding, (\ref{Cf}) was applied to numerical evaluations of $\eta(r)$ and $\gamma(r)$. There is a qualitative similarity between the two complexities, but they are not exactly the same. In section \ref{sectappr} it was mentioned that the step function with the equidistribution of error property is possible the optimal approximation. This is not contradicted by Figure \ref{RQcomp}, which shows that for this family of functions the V-complexity is smaller than the RLE-complexity.\\
A remarkable feature of the RLE-complexity for the family $f_\alpha(x)$ is that the minimum in definition (\ref{Cf}) is achieved in $r=0$ for $0.5\leq \alpha \leq 2$. The compression factor for all these values of $\alpha$ is then $\eta(0)=1$, so no compression, and the error factor $\gamma(0)=1/4$. 

\subsection{GZIP-compression}
Run Length Encoding is not used in practical applications very often. Much more popular are GZIP and similar compression programs. These programs use the Lempel-Ziv 77 (LZ77) encoding algorithm \cite{LempelZiv}, in combination with other tools.\\ The LZ77 algorithm encodes finite strings of symbols from a finite alphabet. It moves from left to right through the sequence and replaces a recurring subsequence by a pointer to the position of the first occurrence of that subsequence. For reasons of efficiency, the algorithm only considers symbols at most $M$ positions back from the symbol currently being processed. As the algorithm moves through the sequence, it forgets data that is too far from the current position.\\
The compression of a sequence by LZ77 is generally shorter than by RLE. For example, the string consisting of $n$ repetitions of $01$ cannot be further compressed by RLE, but has a very short LZ77 code.\\
One notable exception to this rule is the case of {\it monotone} sequences. These are sequences where the symbols are in alphabetical order, according to some ordering of the alphabet of symbols. For instance, $abab$ is not monotone, but $cccabb$ is. For these sequences, RLE and LZ77 are equivalent. There is no extra redundancy for LZ77 to remove beyond the repetitions of the individual symbols.\\ This observation applies to functions $f:[a,b]\rightarrow {\mathbb{R}}$. If the function is monotone increasing or decreasing, then so is its discretization, written as a string. Hence, the RLE-complexity and the GZIP complexity of a monotone function are equivalent.\\
This phenomenon is illustrated in Figure \ref{verga} which shows the sizes of compressed files of the discretization of $f_\alpha(x)=x^\alpha$, for RLE, GZIP and Compress (a Mathematica compression algorithm). Note that for RLE, size means the number of pairs in the encoding, whereas for the other two methods it is the number of bytes in the encoded file. In Figure \ref{vergb}, the graphs are normalized and show the compression methods are indeed very similar for this family of monotone functions.

\begin{figure}[h!]
\centering
\begin{subfigure}{0.45\textwidth}
\centering
\includegraphics[width=\textwidth]{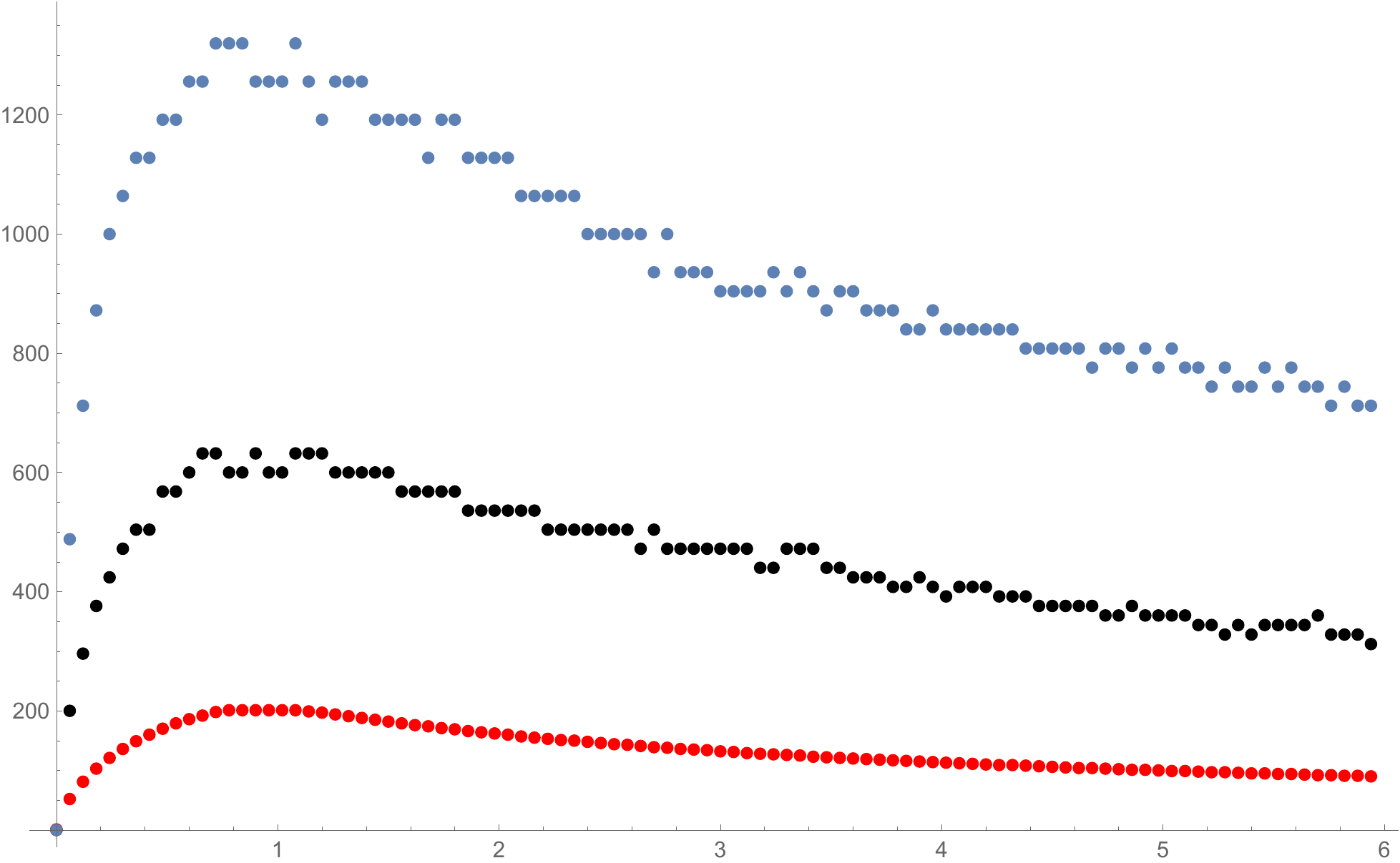}

\caption{File sizes of encoded functions. }
\label{verga}
\end{subfigure}
\begin{subfigure}{0.45\textwidth}
\centering
\includegraphics[width=\textwidth]{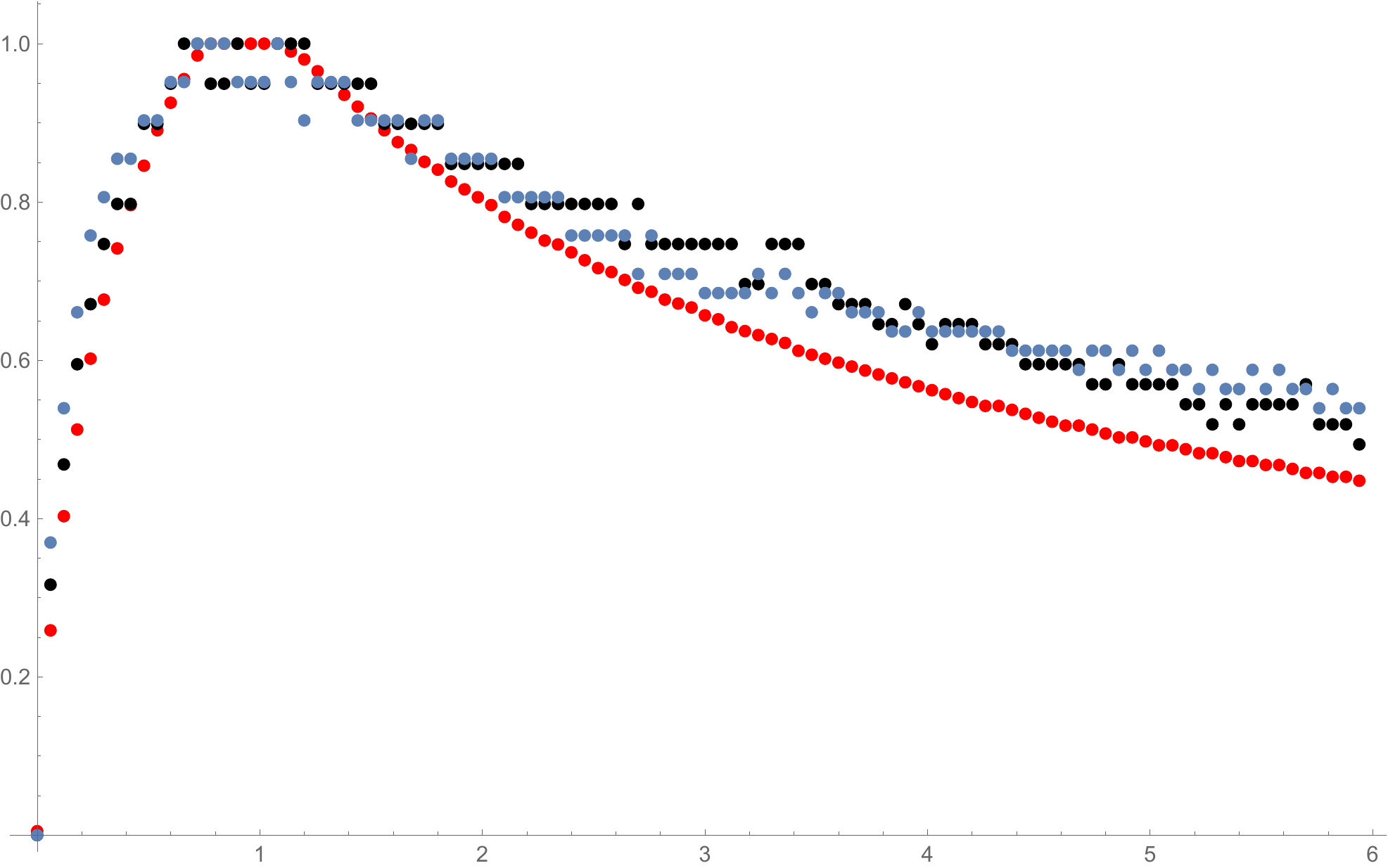}

\caption{Normalized sizes.}
\label{vergb}
\end{subfigure}
\caption{Comparison of RLE, GZIP and Compress encoding of $f(x)=x^\alpha$. The parameter $\alpha$ is on the horizontal axis. The black points correspond to GZIP, blue points to Compress and red points to RLE. For all graphs, $\Delta x=0.005$ and $r=0.8$.}
\label{compLZ77}
\end{figure}
\noindent An important parameter in GZIP complexity is $n/M$, with $M$ the size of the memory window of the GZIP algorithm and $n$ the number of grid points of the discretization. When this ratio is large, the window of the algorithm only covers a part of the domain of $f(x)$. This ratio can be made arbitrarily large by choosing $\Delta x$, or equivalently the approximation error $\varepsilon$, sufficiently small. A relatively short window, as it moves along the sequence, encounters mainly monotone sub sequences.  Hence, if $M$ is kept constant, GZIP encoding again coincides with RLE for large $n$, or equivalently in the limit $\varepsilon \rightarrow 0$.
\begin{figure}[h!]
\centering
\begin{subfigure}{0.45\textwidth}
\centering
\includegraphics[width=\textwidth]{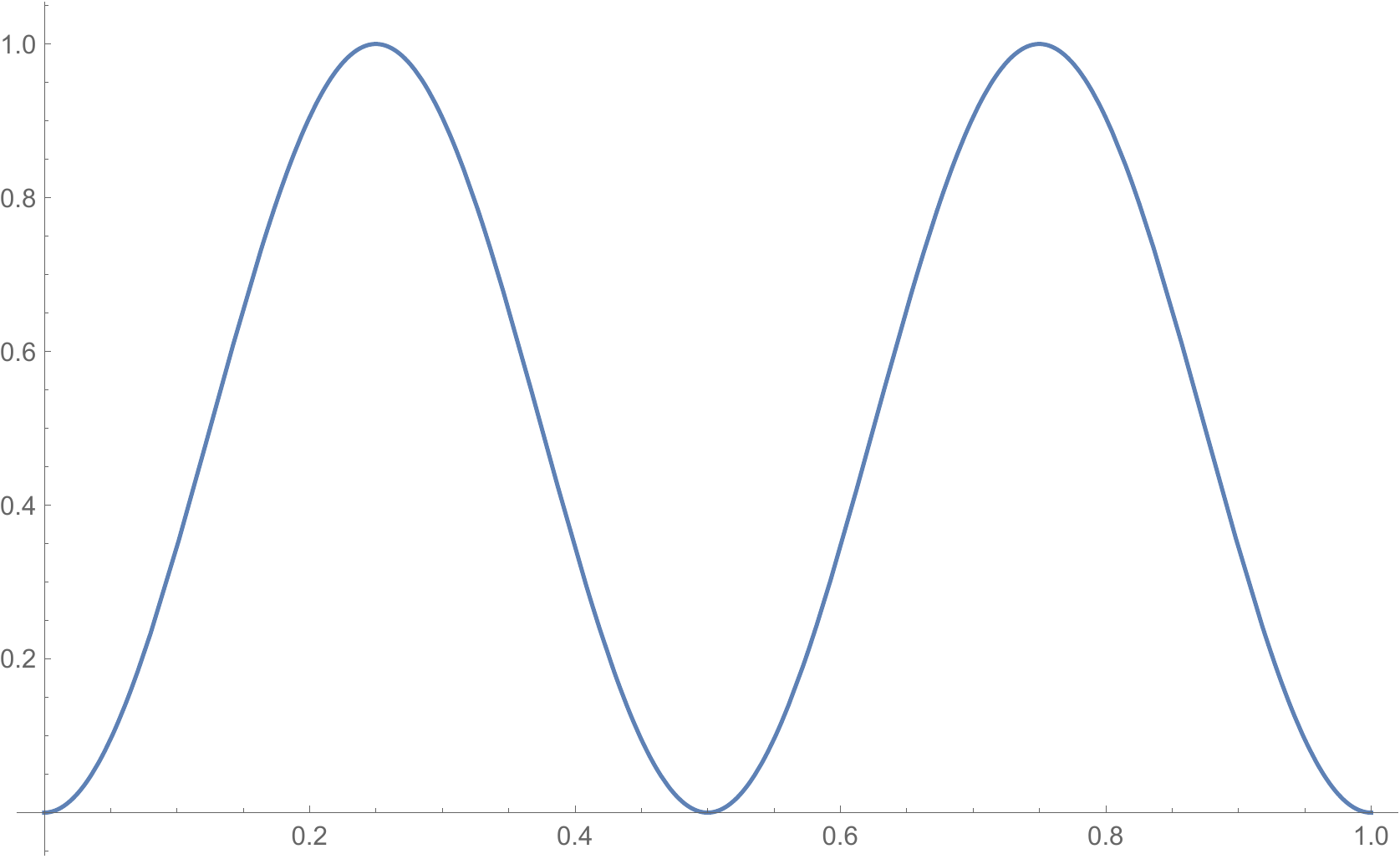}
\caption{$f(x)=\sin^2 \pi x$ }
\end{subfigure}
\begin{subfigure}{0.45\textwidth}
\centering
\includegraphics[width=\textwidth]{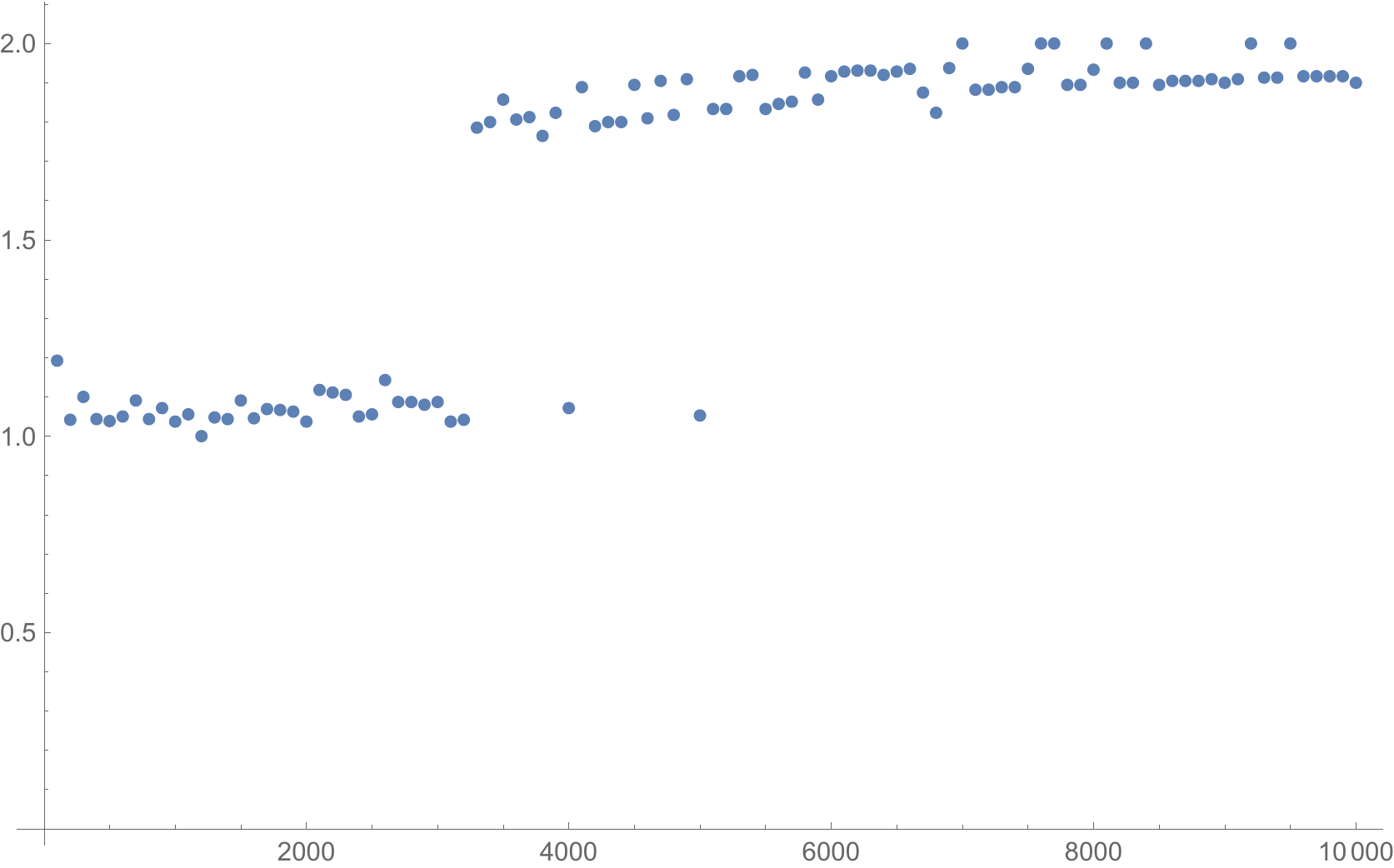}

\caption{Ratio encoding lengths}
\label{ratcomp}
\end{subfigure}
\caption{Ratio of the encoding length of $f(x)=\sin^2 \pi x$ on the interval $[0,1]$ to encoding length over $[0,1/2]$. Horizontally are the number of grid points. $M=4096$ and $r=0.5$.}
\label{ratlen}
\end{figure}
\noindent The above phenomenon is illustrated in figure \ref{ratlen}. The function $f(x)=\sin^2 \pi x$ is encoded and compressed using GZIP, first for the interval $[0,1/2]$ and then for the interval $[0,1]$. The lengths of these compressed files are $l_1$ and $l_2$, respectively. This is repeated for values of $n$ ranging from $n=100$ to $n=10^4$. In figure \ref{ratcomp}, the ratios $l_2/l_1$ are plotted. For $n<M$, this ratio is close to 1. This is because the LZ77 algorithm uses the fact that $f(x)$, restricted to $[1/2,1]$, is a translated copy of the function restricted to $[0,1/2]$. Therefore, the LZ77 code for $f(x)$ on the full interval $[0,1]$ is equal to the code for the function on $[0,1/2]$ plus one pointer to that code. For these values of $n$, LZ77 can fully exploit the redundancy of the function.\\
For $n>M$, however, the window of the LZ77 algorithm cannot hold the full function in its memory. The output of the code will be two copies of the code for the function on $[0,1/2]$ and the ratio $l_2/l_1$ is close to 2. \\
The upshot of this section is twofold. First, it is hypothesized that for a differentiable function $f$, the V-complexity is close to the RLE-complexity of the function. Second, the RLE-complexity is equivalent to the GZIP-complexity, in the limit that the grid size, and therefore the $L_1$ error, goes to zero.

\section{Complexity in a cup}
\label{CCex}

 Consider a cup of cappuccino. In its initial state, the top half of the cup consists of only cream whereas the bottom half is filled with just coffee. On the micro level of the individual particles, the milk and coffee molecules start a random walk, occasionally colliding with each other. This level can not be observed. On the macro level the randomness of the motion of the individual particles has become invisible. What can be observed is the color of the liquid, equivalent with the local density of the cream molecules. The color-distribution evolves continuously in time and space.\\ 
 As time progresses, the color near the boundary between cream and coffee will develop a gradient going from white, through brown to black. This boundary layer will increase in size until all the liquid has the same light brown color. This is the equilibrium state.\\ 
Heuristically, the evolutions in time of complexity and entropy go as follows. At the macro level, the initial state is easily described, as is the final  equilibrium state. For times between between these extremes, however, the state is more complicated. The information-theoretic complexity of the macro level of the system, therefore, goes from almost zero to some maximum and then back to zero as the equilibrium is approached. In contrast, the entropy of the system, which is calculated using the probability distribution of the positions of the individual particles, starts at a relatively low level and then increases monotonically.\\ 
To keep the models as simple as possible, it will be assumed that only diffusion is at work. The initial condition is homogeneous in the $x$ and $y$ directions, so the only dimension that matters is the $z$ component of the position of the particles. While individual particles of coffee or cream can move in the horizontal directions, their averaged motion over time in these directions is zero. The following models are therefore one-dimensional in space.

\subsection{Lattice gas automaton}
 The simplest cellular automata to simulate molecular dynamics are lattice gas automata \cite{Chopard}. 
\begin{figure}[h!]
\centering
\includegraphics[scale=0.8]{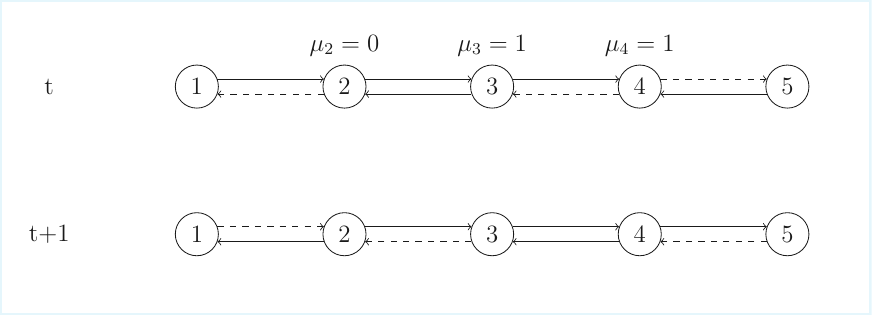}
\caption{The lattice gas model. Dashed lines represent coffee particles, solid lines are cream particles.}
\label{latga}
\end{figure}
\noindent A one-dimensional lattice gas model of the coffee cup can be imagined as a two lane highway. At every $k=1,2, \ldots N$ kilometers there is a roundabout, henceforth called a node.  Particles on the top lane move to the right, and on the bottom lane to the left . All particles move at the same constant speed. Every lane segment contains exactly one particle, either coffee or cream. In the top picture of Figure \ref{latga}, a possible situation for $N=5$ is shown at a time $t$. Note that there are no incoming particles from the left (right) at end-nodes $k=1$ ($k=5$).\\
At time steps $t=1, 2, \ldots$, two particles enter a node, one from the left and the other from the right. At the endpoints only one particle enters. This entering happens simultaneously for all nodes. At every node there is a traffic light which is either red or green. If the light is green, both particles continue in the same direction as they were coming in. They are said to pass each other. If the light is red, both particles turn around and travel back from where they came, on the opposite lane. This situation is called a collision. The end-nodes do not have a traffic light. There, particles will just return in the direction from where they came.\\
The color of the traffic light is a random variable $\mu_k(t)$ which takes on the value $\mu_k(t)=1$ (green) with probability $1/2$ and $\mu_k(t)=0$ (red), also with probability $1/2$. The variables $\mu_k(t)$ are mutually independent, for all $k$ and $t$.\\
Figure \ref{latga} shows the dynamics during one time step. When two particles of the same type enter a node, they pass each other, independent of the value of $\mu$. This happens at nodes $k=2$ and $k=4$. An alternative explanation is that the particles collide, but the outcomes in both interpretations is the same, as the particles are considered to be anonymous. When two particles of a different type enter a node, they collide with probability $1/2$ or pass, also with probability $1/2$. This happens at node $k=3$. In this case, the outcome would have been different in the case that $\mu_3=0$.\\
\begin{figure}[h!]
\centering
\begin{subfigure}{0.4\textwidth}
\begin{minipage}{\columnwidth}
\includegraphics[width=4.5 cm]{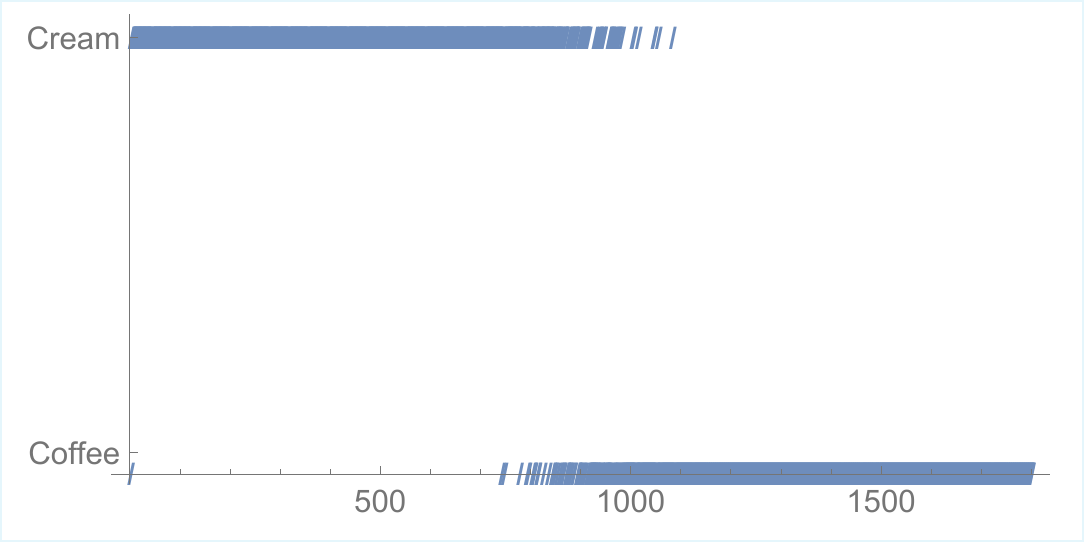}
\end{minipage}
\begin{minipage}{\columnwidth}
\includegraphics[width=4.5 cm]{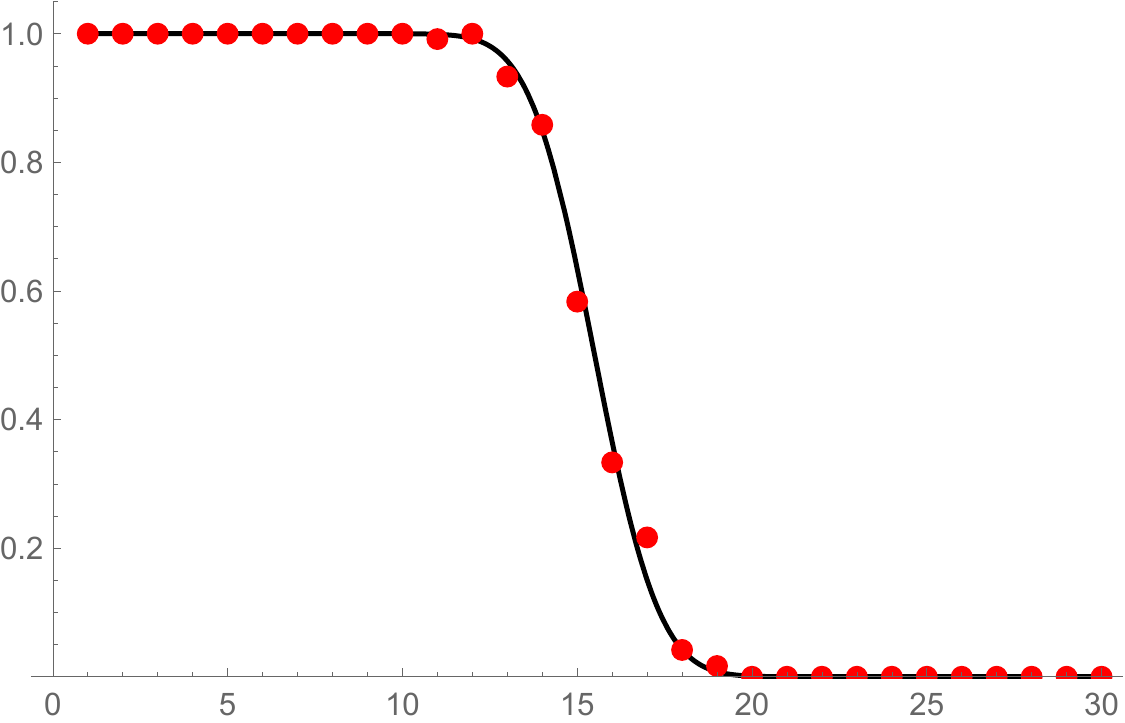}
\end{minipage}
\caption{$t=0.003\,T$.}
\end{subfigure}
\begin{subfigure}{0.4\textwidth}
\begin{minipage}{\columnwidth}
\includegraphics[width=4.5 cm]{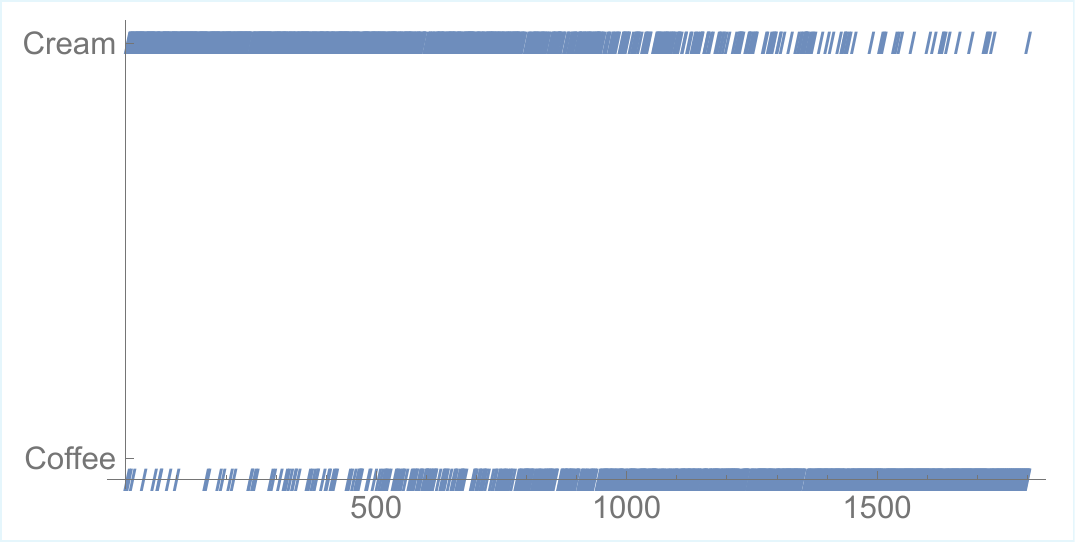}
\end{minipage}
\begin{minipage}{\columnwidth}
\includegraphics[width=4.5 cm]{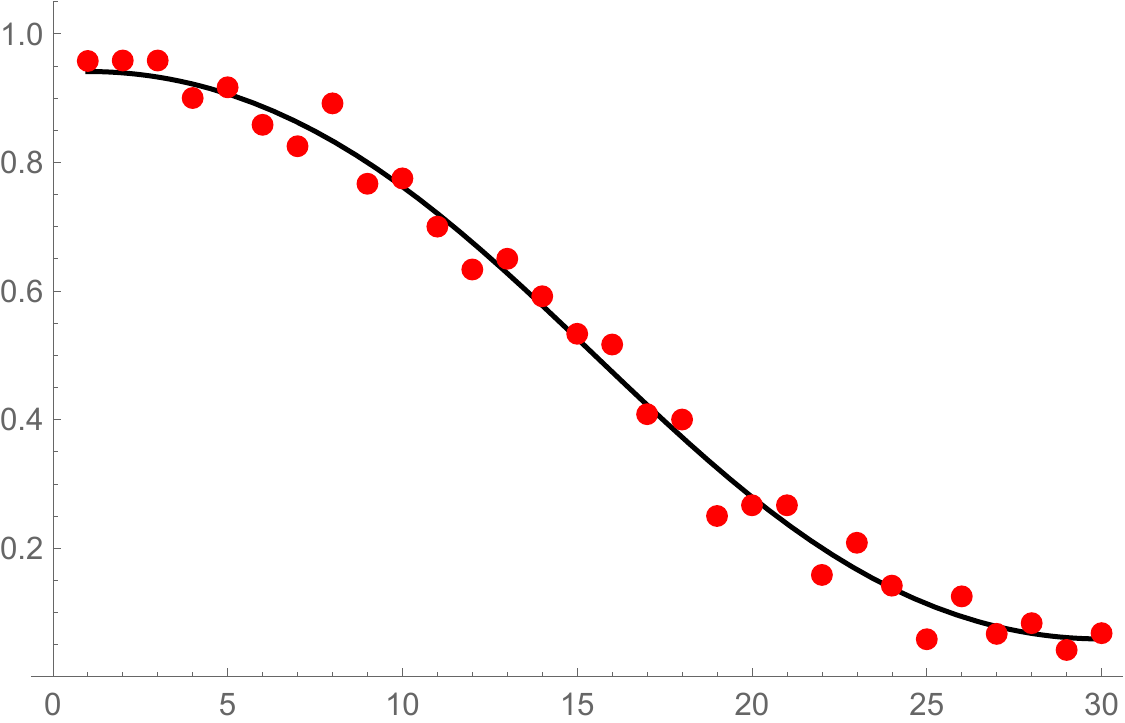}
\end{minipage}
\caption{$t=0.07\,T$.}
\end{subfigure}
\begin{subfigure}{0.4\textwidth}
\begin{minipage}{\columnwidth}
\includegraphics[width=4.5 cm]{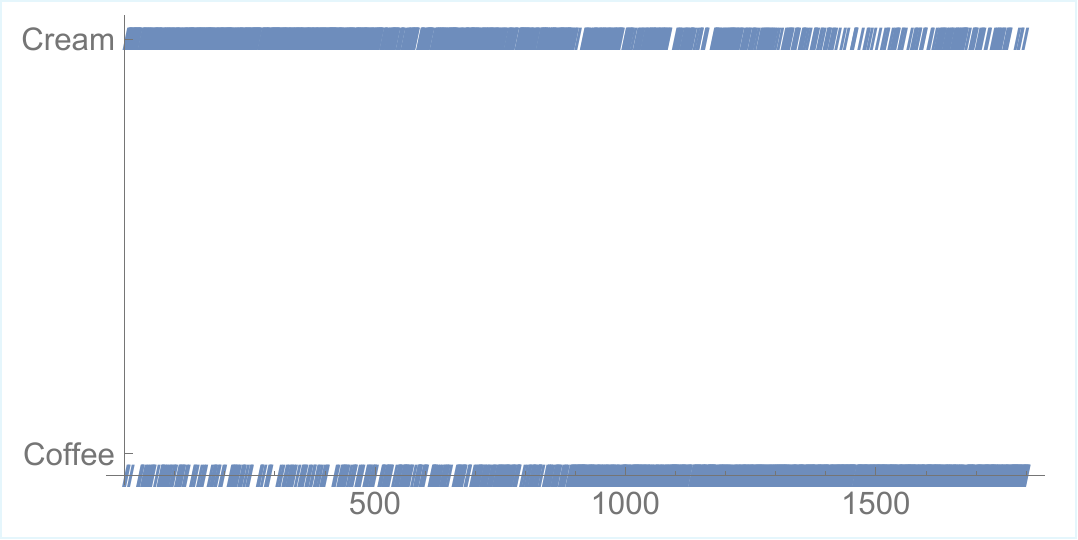}
\end{minipage}
\begin{minipage}{\columnwidth}
\includegraphics[width=4.5 cm]{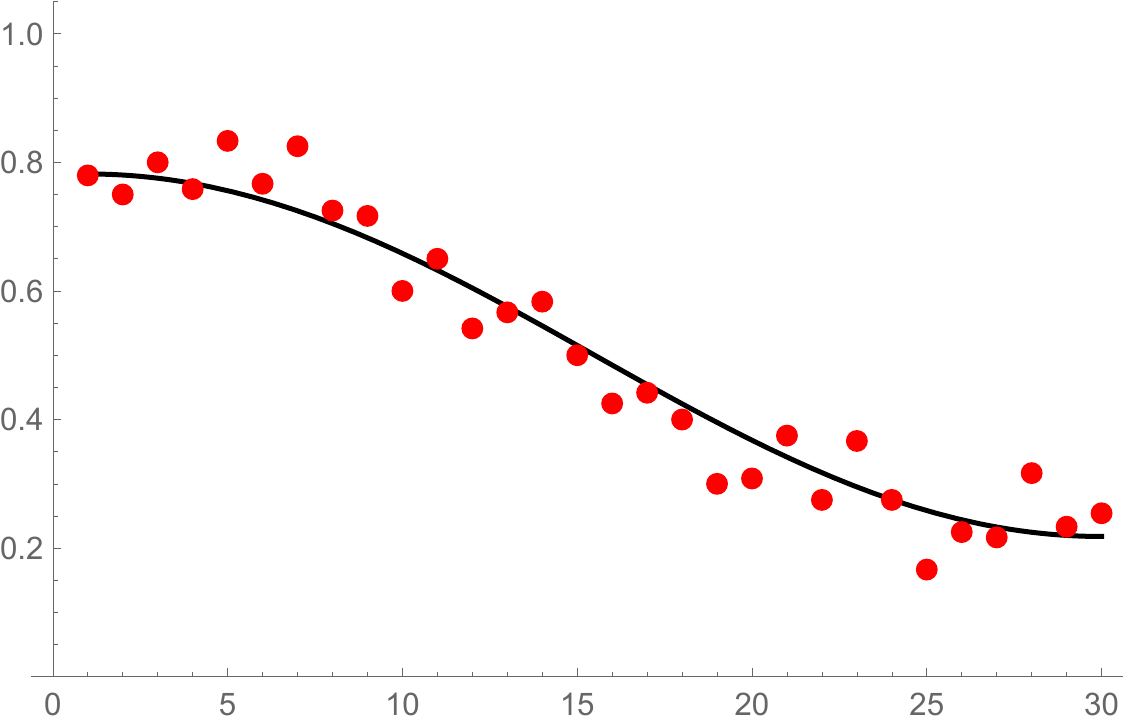}
\end{minipage}
\caption{$t=0.17\,T$.}
\end{subfigure}
\begin{subfigure}{0.4\textwidth}
\begin{minipage}{\columnwidth}
\includegraphics[width=4.5 cm]{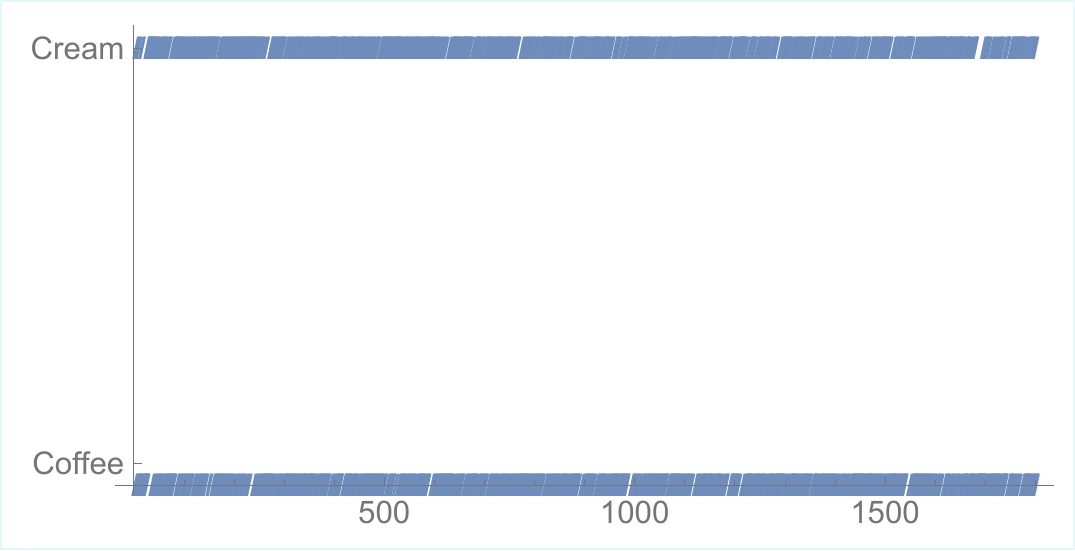}
\end{minipage}
\begin{minipage}{\columnwidth}
\includegraphics[width=4.5 cm]{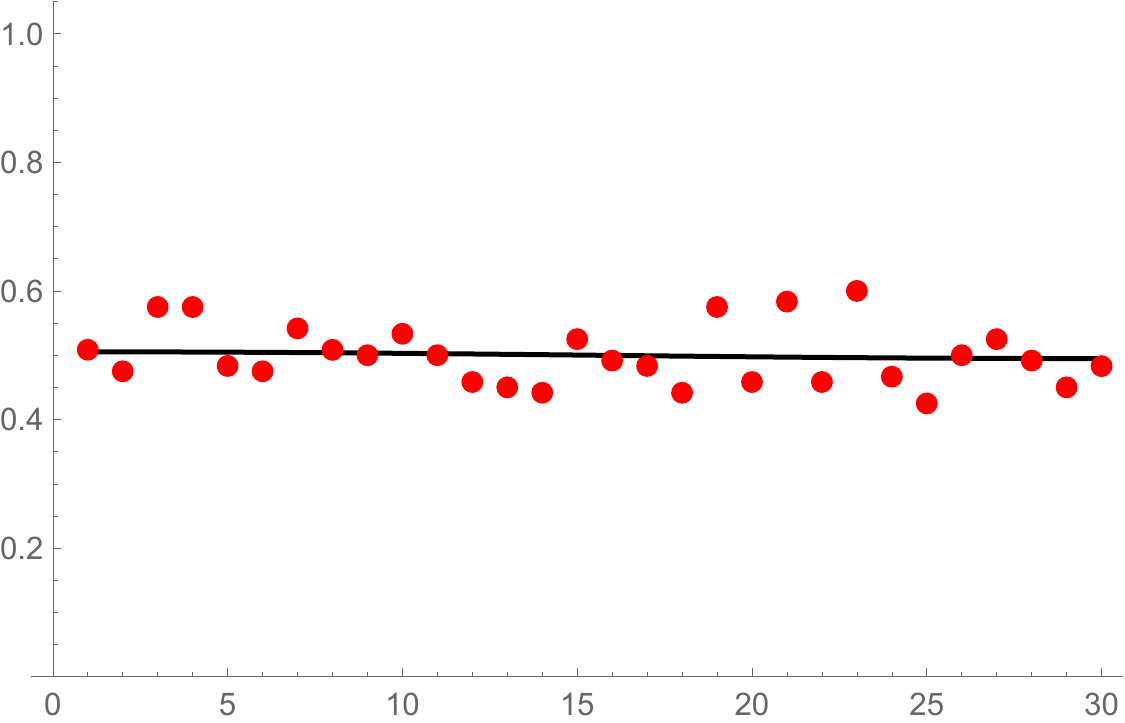}
\end{minipage}
\caption{$t=\,T$.}
\end{subfigure}
\caption{State of the automaton at four different times. Parameters are $N=1800$, $C=30$, $B=30$ and $T=3.24 \times 10^6$. Top graph shows which nodes have a cream/coffee particle entering from the left. The bottom graph shows average color of the cells. The black curve is the solution of (\ref{diff}).}
\label{densdis}
\end{figure}
\noindent The parameters of the system are the number of nodes, $N$ and the total number of time steps $T$. Figure \ref{densdis} shows the result of a simulation with $N=1800$, $T=3.24 \times 10^6$ and an initial condition where for the first 1000 nodes all lanes between them are occupied by cream particles and for the last 1000 nodes with coffee. The top graph of each panel shows the nodes where a cream/coffee particle enters from the left. The picture for the nodes where they enter from the right is much the same and is not shown here.
\begin{figure}[h!]
\centering
\includegraphics[scale=0.5]{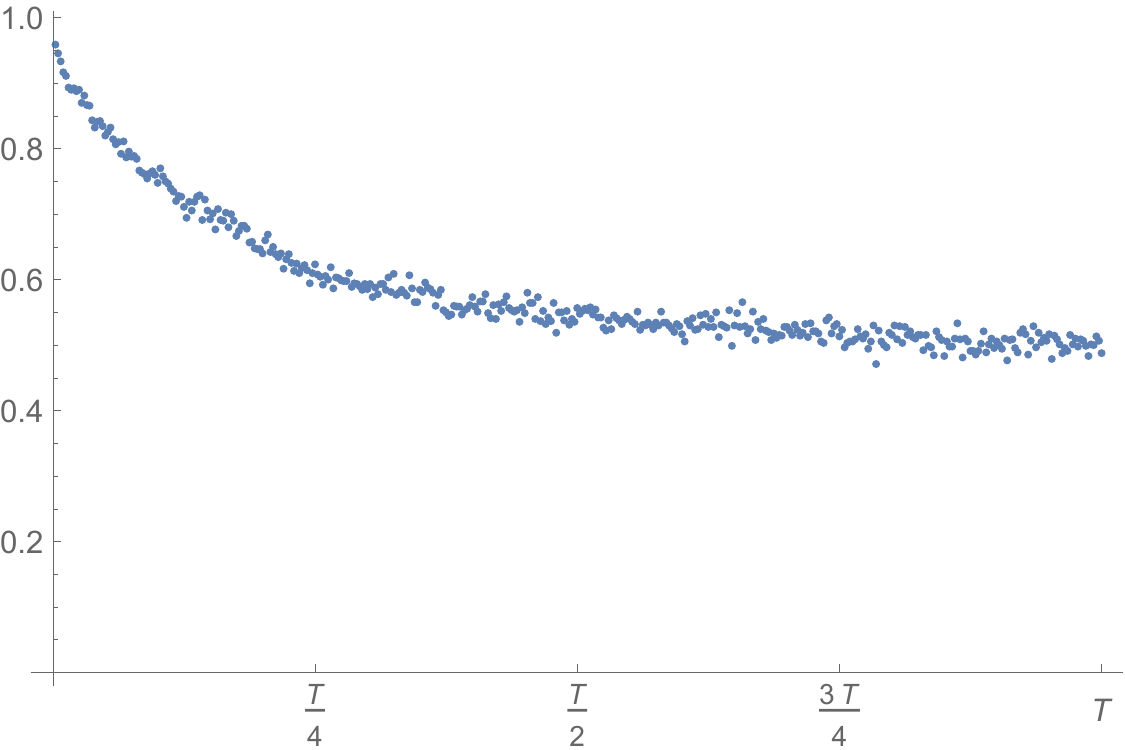}
\caption{Fraction of cream particles in upper half of the cup. $T=3.24 \times 10^6$}
\label{fracwhite}
\end{figure}
Figure \ref{fracwhite} shows the fraction of cream particles in the top half of the cup. This fraction starts at 1 and then decreases to $1/2$, indicating that the cream, and thus also the coffee, has spread out over the whole cup.\\
The color of a lane segment is defined as the average value of the particles traveling on that segment, both going left and going right, with cream$=1$ and coffee$=0$. The color of a lane segment is therefore equal to $1$ for two cream particles, $1/2$ for a coffee/cream combination or $0$ for two coffee particles.

\subsection{Apparent Complexity}

\begin{figure}[h!]
\centering
\includegraphics[width=0.8\textwidth]{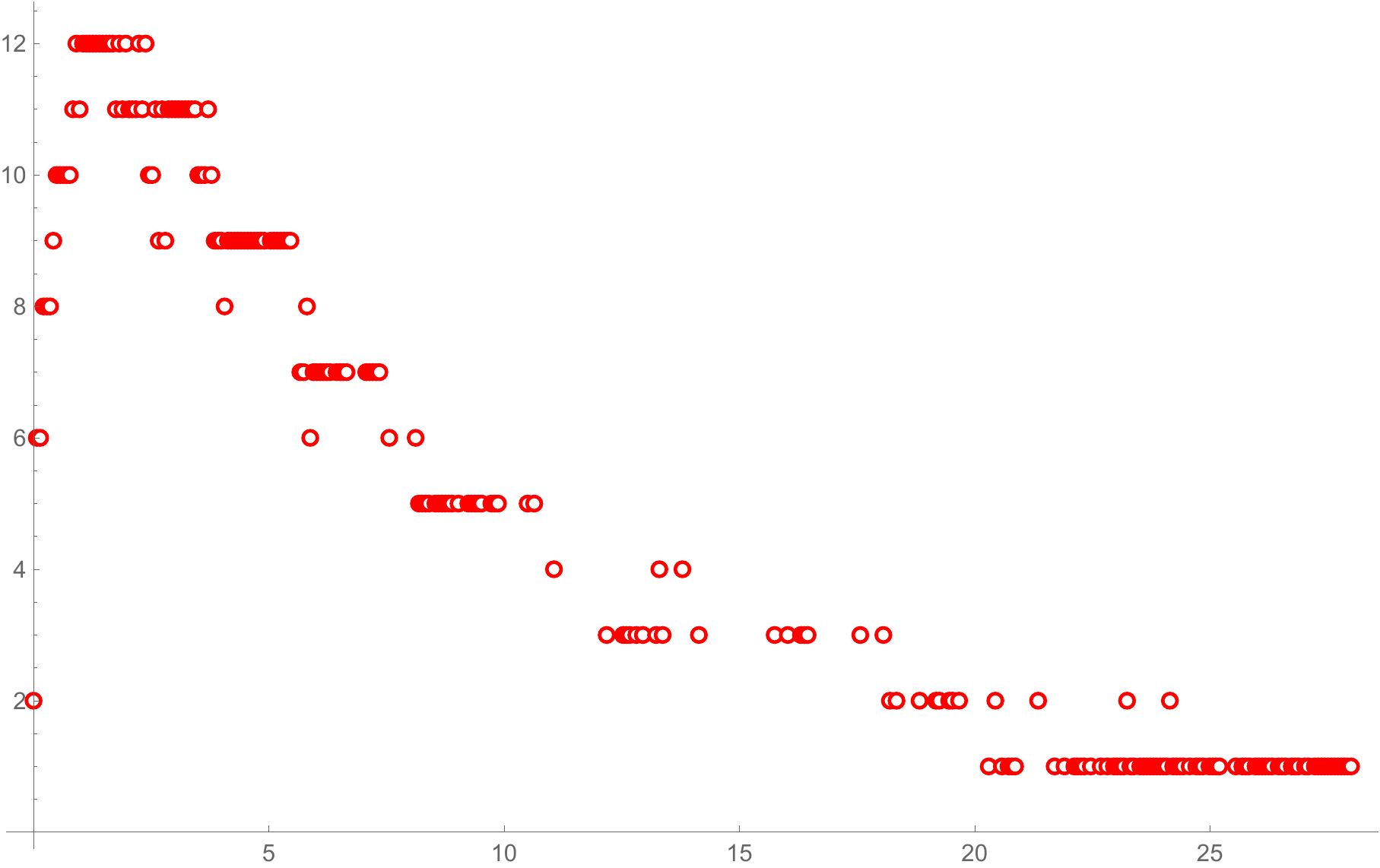}
\caption{Apparent complexity of lattice gas automaton. The compression algorithm used is RLE. Horizontal are time steps in units of $10^4$.}
\label{simulcomp}
\end{figure}
\noindent The definition of Apparent Complexity is given in \cite{Carroll}, in particular for a $n$-dimensional grid with a $0$ or a $1$ on every grid point. The grid is first coarse grained, by replacing the values on each grid points by the average value of nearby grid points. Then, the coarse grained values are rounded to the closest multiple of some $\Delta y$. Finally, the binned grid is compressed using the GZIP algorithm. The length of this file is the apparent complexity of the grid.\\
Note that this definition is similar to Efficient Complexity. By the process of coarse graining, Apparent Complexity gives an explicit method to remove the incidental or random parts of the system. Compression by the GZIP algorithm is a practical implementation of the "concise description" used in the definition of Efficient Complexity.\\
Apparent Complexity can be applied to the lattice gas automaton, with the difference that the grid points are the colors of the line segments and they can have values $0$, $1/2$ or $1$.\\
To calculate the apparent complexity of the model, it was run with $N=800$ and $T=2.88\,.10^5$. The 800 segments were partitioned in 20 groups of each 40 segments and the results in each group were averaged, giving 20 values. The range of the outcomes is $[0,1]$. This interval was partitioned into 12 bins. At every timestep, therefore, the outcome was a list of 20 integers with values between 1 and 12. The simulation was run 20 times and the results averaged. The effect of averaging all simulations is that the averaged lists, for all times, are perfectly ordered, from large to small. In other words, the  averaged lists are monotonic. This is to be expected, since it is extremely unlikely that a layer in the cup is lighter than a layer above it. This also implies that for the calculation of apparent complexity, either RLE or GZIP encoding can be used since these measures are equivalent for monotonic strings. Here, the choice is for RLE.\\
The initial condition is equivalent to the list consisting of the value 12 (the most white) repeated 10 times followed by the value 1 (the most black), also 10 times. The RLE of the initial condition is therefore $\{\{12,10\},\{1,10\}\}$ and the  apparent complexity is 2. As time progresses, more bins appear in the list of outcomes. At a time corresponding to Figure (\ref{densdis})(b), all 12 bins appear in the list of outcomes and the apparent complexity is equal to 12. For larger times, as the distribution becomes flatter, the RLE decreases again. At very large times, all outcomes are in the middle bin, labeled 6, and the apparent complexity equals 1.\\
The  conclusion that can be drawn from these simulations is that apparent complexity performs as expected for the lattice gas automaton model.

\subsection{The diffusion equation}
\begin{figure}[h!]
\centering
\includegraphics[scale=0.35]{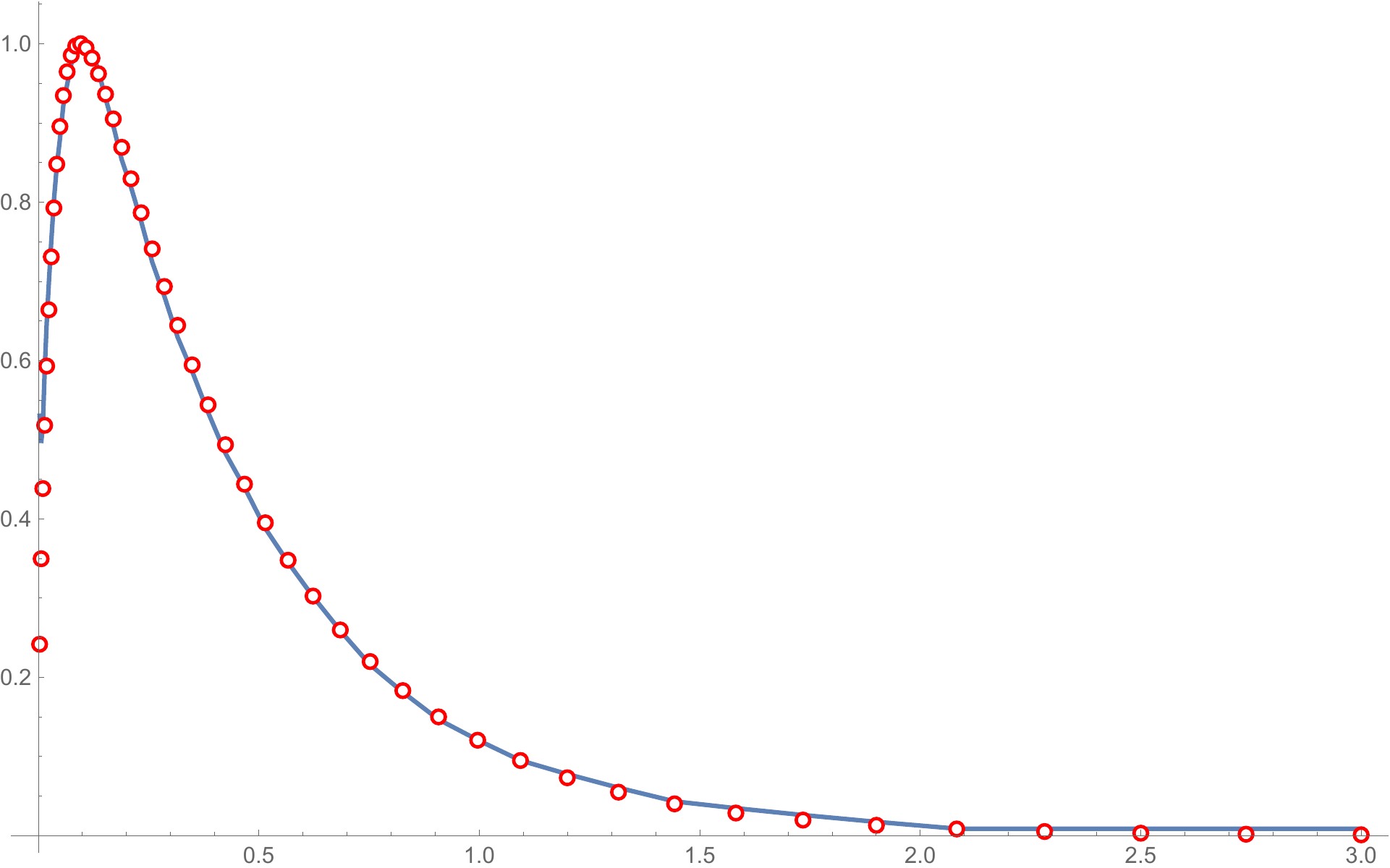}
\caption{V-complexity of the solution of (\ref{diff}) (solid line) and RLE-complexity of discretized solution (red circles). Both complexities are normalized. The horizontal axis displays time.}
\label{diffcomp}
\end{figure}
\noindent The expected value of the color of a specific segment at a specific time, rather than its actual value for a particular run, is a real number somewhere in the interval $[0,1]$. It will be called the density of the color at that position and time.\\
In \cite{Chopard} it is shown that the density of  the color satisfies the classic partial differential equation for diffusion as $N\rightarrow \infty$. In particular, the spatial lattice $\{1,2,\ldots, N\}$ can be mapped to the interval $[-1,1]$, giving a lattice spacing $\lambda=2/N$. The values $-1\leq x\leq 0$ correspond to the upper half of the cup and $0\leq x\leq 1$ to the lower half. Time is mapped from $\{1,2,\ldots, T\}$ to an interval $[0,c]$ by using a time step $\tau=c/T$, where the value of $c$ is yet to be decided. Keeping the ratio $\lambda^2/\tau$ constant, while letting both $\lambda$ and $\tau$ go to zero, yields an equation for the density $u(x,t)$ of the cream particles:
\begin{align}
\label{diff}
    \frac{\partial u}{\partial t}&= D\frac{\partial^2 u}{\partial x^2}\, , 
\end{align}
with $D=\frac{\lambda^2}{2 \tau}$. For convenience, $c=\frac{2T}{N^2}$ is chosen, with the values of $T$ and $N$ as in Figure \ref{densdis}, which makes $D=1$.\\
The boundary values are of Neumann type, because the particles bounce off of the end nodes. Together with the initial condition, this yields.
\begin{align}
\label{bound}
    \quad u_x(-1,t)=u_x(1,t)&=0 \nonumber \\
    u(x,0)&=\begin{cases} 
    1 & \text{if } -1 \leq x <0 \\
    1/2  & \text{if } x=0\\
   0 & \text{if } 0 < x \leq 1  \, .
        \end{cases}
\end{align}
Using a Fourier series expansion, the solution of (\ref{diff}) with the initial and boundary conditions (\ref{bound}) is given by
\begin{align}
\label{uxt}
    u(x,t)=1/2-(2/\pi)\sum_{k=1}^{\infty}\frac{e^{-((2k-1)\pi/2)^2 t}}{2k-1}\sin((2k-1)\pi x/2)\, .
\end{align}
The V-complexity of $u(x,t)$ can be calculated numerically for every value of $t$, using a finite number of terms from (\ref{uxt}). The result is shown in Figure (\ref{diffcomp}), where the outcomes are normalized such that the maximum value of the V-complexity equals 1. To confirm the hypothesis of section \ref{Qandcomp}, the similarly normalized RLE-complexity of the discretized $u(x,t)$ is also plotted, with $r=0.8$. As can be seen, they are practically identical, although this depends on the choice of $r$.\\
The similarity between Figures (\ref{simulcomp}) and (\ref{diffcomp}) is striking and it illustrates that the two approaches to the "complexity in a cup" paradigm lead to the same result. This is explained as follows. The cellular automaton model, after coarse graining and in the limit that $\Delta x\rightarrow 0$, is equivalent to the diffusion equation. Both approaches then apply a measure of complexity. In the case op Apparent Complexity, it is the length of the GZIP encoded output string. This length grows without bound as $\Delta x\rightarrow 0$. Rescaling this measure by $\Delta x\rightarrow 0$ gives the compression complexity, as defined in the previous section. In the case of the diffusion equation, V-complexity is used on the solution of the equation. According to the hypothesis in section \ref{Qandcomp}, GZIP compression complexity and V-complexity are equivalent. So, the Apparent Complexity of the cellular automaton and the (modified) Effective Complexity of the diffusion equation model are equivalent.

\section{Conclusion}
\label{Concl}
Wittgenstein's famous proposition: "everything that can be said, can be said clearly"  can be paraphrased as a motto for information-theoretic complexity as: "everything that can be described, can be described concisely". This paper defines {\it{how}} concisely, for the case of differentiable functions on an interval.\\
Descriptions need a vocabulary. Here, this means a set of functions in terms of which every differentiable function can be expressed. The choice is to take the piecewise constant functions, also known as step functions, as the basic functions. The graphs of these functions show sharp contrasts, such as the one in Figure \ref{d0r45}. Their gray scale pictures resemble one-dimensional Mondriaan paintings, albeit with a rather dull color palette. To the eye, these pictures appear simple, easy to describe. On the other hand, a graph such as in Figure \ref{d2r45} contain a lot of fuzziness and would take more words to describe.\\ 
Any differentiable function can be approximated by a step function, but only to a certain accuracy. By taking approximating step functions with more and more steps, the accuracy becomes ever better. The number of steps and the accuracy of approximations are the ingredients of the main results of this paper. The complexity of a differentiable function can be defined as the smallest ratio of the number of steps in a step function approximation to the accuracy of the approximation, as the accuracy nears perfection. This ratio was named the V-complexity of the function, a nod to its basis in Visual perception. Also, there exists a concrete expression for the V-complexity of a function, in terms of the its derivative.\\ 
If the V-complexity is small, then a step function with a small number of steps suffices to describe the function to a high accuracy. It is a function with fairly sharp contrasts. The function is easy on the eyes, in a quite literal sense.\\
\\
Whereas V-complexity is based on the appearance of the graph of a function, a seemingly different approach is be based on the {\it{compressibility}} of the function. A simple function, in this approach, would be one whose approximation can be compressed into a short file. Surprisingly, these two different approaches lead to qualitatively equivalent results.\\\\
V-complexity can be used as an ingredient in the characterization of dynamical systems with many interacting elements, also known as Complex Systems. The Effective Complexity of such systems is defined as "the length of a concise description of the system's perceived regularities" \cite{GellMann}. If the description of the regularities is a differentiable function on an interval, its description generally has an infinite length. It then makes sense to amend the definition of Effective Complexity into "the V-complexity of the function $f$ which describes a system's perceived regularities".\\
This idea is applied to the example of diffusion of cream in a cup of coffee. It is shown that using the diffusion equation and the amended version of Effective Complexity leads to a result equivalent to starting with a cellular automaton, course graining and applying compression. This latter approach is called Apparent Complexity \cite{Carroll}. The reason is that V-complexity and compression complexity are equivalent and that the automaton is described by the diffusion equation, in the limit that the grid size goes to zero.\\
\\
Much work is left to be done in perfecting the concept of V-complexity. Some statements made here are actually hypotheses, although supported by simulations. This holds for the assumption that the minimal number of steps in the step function approximation is close, or even equal, to the number found using the equidistribution of errors principle. Also, the equivalence of V-complexity and compression complexity needs further clarification.\\
V-complexity can possibly be generalized to functions of several variables, which would be very useful for the study of Complex Systems with more than one spatial variable.\\
This paper, then, is hopefully just the start of characterizing Complex Systems which are described in terms of differentiable functions.

\newpage

\section*{Appendix A} 
In the following, it is assumed that there are no subintervals $[c,d] \subset [a,b]$ such that $f(x)=\alpha$ for all $x \in [c,d]$. This is not a restrictive assumption, since such an interval where $f(x)$ is constant can be removed by defining ${\hat{f}}(x)=f(x)$ for $a\leq x \leq c$ and ${\hat{f}}(x)=f(x+d-c)$ for $c\leq x \leq b+c-d$. For a given $\varepsilon >0$ the optimal approximation of $f$ is then equal to that of ${\hat{f}}$, with the addition of the interval $([c,d],\alpha)$. In this case, the value of $N(\varepsilon)$ for $f$ is increased by one, compared to that of ${\hat{f}}$. This increase becomes asymptotically insignificant when $\varepsilon \rightarrow 0$.\\ \\
Recall that, for a $f:[a,b]\rightarrow {\mathbb{R}}$ and a given $\varepsilon >0$, the stepfunction approximation consists of a set $\{\{I_1,g_1\}, \ldots,\{I_N,g_{N(\varepsilon)}\}\}$. It will be assumed that for small values of $\varepsilon$, the length of the intervals $I_k$ is also small. On an interval $I_k=(x_{k-1},x_k]$, the function $f(x)$ can then be approximated by
\begin{align}
\label{Taylor}
f(x)\approx f(x_{k-1})+f'(x_{k-1})(x-x_{k-1}) \, .
\end{align}
The optimal value for $g_k$ follows from
\begin{lemma}
Let $l(x)=\alpha x+\beta$ for $x\in[c,d]$. The value of $g$ that minimizes
\begin{align}
\label{err}
\int_c^d |l(x)-g|\,\textrm{d}x
\end{align}
is
\begin{align}
\label{qzero}
    g=\alpha(c+d)/2+\beta=l((c+d)/2) \, ,
\end{align}
and the corresponding error is 
\begin{align}
\label{perr}
    ||l-g||_1=(1/4)|\alpha|(d-c)^2\, .
\end{align}
\end{lemma}
$\textbf{Proof:}$ The statement is true for $\alpha=0$, so assume $\alpha > 0$. The case $\alpha <0$ is equivalent. Substituting $l(x)=\alpha x+\beta$ in (\ref{err}), and writing
\begin{align}
    \tilde{g}=\frac{g-\beta}{\alpha} \,,
\end{align} yields
\begin{align}
\label{ler}
\int_c^d |l(x)-g|\,\textrm{d}x&=\alpha \int_c^d |x-\tilde{g}|\,\textrm{d}x\nonumber \\
&= \alpha\int_c^{\tilde{g}} (\tilde{g}-x)\,\textrm{d}x+\int_{\tilde{g}}^d (x-\tilde{g})\,\textrm{d}x \nonumber \\
&=\alpha  \left( (\tilde{g}-c)^2+(d-\tilde{g})^2\right)/2\, .
\end{align}
Differentiating (\ref{ler}) with respect to $\tilde{g}$ shows that $\tilde{g}=(c+d)/2$ is a minimum, which corresponds to (\ref{qzero}). Substituting $\tilde{g}$ in (\ref{ler}) gives (\ref{perr}). $\qed$\\ \\
{\bf{Proof}} of Theorem \ref{thm1}.\\
Using the approximation (\ref{Taylor}), the expression for the error (\ref{perr}), and $\alpha=f'(x_{k-1})$, the equidistribution condition (\ref{equi1}) applied to $I_k=(x_{k-1},x_k]$ becomes
\begin{align}
   (1/4)|f'(x_{k-1})|(x_{k}-x_{k-1})^{2}=\varepsilon/ N(\varepsilon) \, .
\end{align}
This yields a recursion for the endpoints $x_k$ of the intervals $I_k$, starting with $x_0=a$:
\begin{align}
\label{recur}
    x_{k}-x_{k-1}=\left(\frac{4 \varepsilon}{|f'(x_{k-1})|N(\varepsilon)}\right)^{1/2}\, , \quad x_0=a\, .
\end{align}
The solution of (\ref{recur}) depends on $N(\varepsilon)$ which at this point is still unknown. It can be found using the second equidistribution condition (\ref{equi2}), which becomes
\begin{align}
\label{closure}
    x_{N(\varepsilon)}=b \, .
\end{align}
To solve (\ref{recur}) and (\ref{closure}), an approximation to $x_k$, $k=0,1,\ldots, N(\varepsilon)$ will be made, taking advantage of the fact that (\ref{recur}) can be transformed to a differential equation as $\varepsilon \rightarrow 0$.\\
Define 
\begin{align}
    y(s)=x(\lfloor N(\varepsilon) s \rfloor)\, , \quad 0 \leq s \leq 1 \, ,
\end{align}
where $\lfloor . \rfloor$ is the floor function and $x(k) \equiv x_k$.
Conversely,
\begin{align}
\label{xfuny}
    x(k)=y(k/N(\varepsilon))\, , \quad k=0,1,\ldots N(\varepsilon)\, .
\end{align}
Define 
\begin{align}
\label{delts}
   \Delta s(\varepsilon)=1/N(\varepsilon)\, . 
\end{align} Then
\begin{align}
    x(k)-x(k-1)&=y(k/N(\varepsilon))-y(k/N(\varepsilon)-\Delta s(\varepsilon)) \nonumber\\
    \label{ysk}
    &=y(s_k)-y(s_k-\Delta s(\varepsilon)) \, .
\end{align}
Substituting (\ref{recur}) in (\ref{ysk}), dividing by $\Delta s(\varepsilon)$ and using (\ref{delts}) yields
\begin{align}
\label{preydif}
    \frac{y(s_k)-y(s_k-\Delta s(\varepsilon))}{\Delta s(\varepsilon)}=\left(\frac{\varepsilon} {\Delta s(\varepsilon)} \right)^{1/2}\left(\frac{4}{|f'(s_k)|}\right)^{1/2}\, .
\end{align}
This equation will only lead to a non-trivial differential equation if.
\begin{align}
    \frac{\varepsilon} {\Delta s(\varepsilon)}={\cal{O}}(1)\, \quad {\textrm{as} \,\,\varepsilon \rightarrow 0} \, .
\end{align}
Specifically, define $V(f)$ through:
\begin{align}
\label{defdels}
    \Delta s(\varepsilon)=\frac{\varepsilon}{V(f)} \quad  \leftrightarrow \quad N(\varepsilon)=V(f) \varepsilon^{-1}
\end{align}
With this choice and taking $\varepsilon \rightarrow 0$, (\ref{preydif}) becomes 
\begin{align}
\label{dvy}
    y'(s)=V(f)^{1/2}\left(\frac{4}{|f'(s)|}\right)^{1/2} \, , \quad y(0)=a \, .
\end{align}
This can be rewritten as an integral equation:
\begin{align}
\label{inty}
    \int_a^{y(s)}|f'(q)|^{1/2}\, \textrm{d}q=\left( 4V(f)\right)^{1/2}\, s \, .
\end{align}
In terms of $y(s)$, the second equidistribution equation (\ref{closure}) is \begin{align}
\label{y1b}
    y(1)=b \, .
\end{align}
Substituting (\ref{y1b}) in (\ref{inty})  then gives
\begin{align}
    V(f)=\frac{1}{4}\left(\int_a^b|f'(x)|^{1/2}\textrm{d}x\right)^{2}\,.
\end{align}
This proves part (a) of Theorem 1.\\ \\
The solution of (\ref{dvy}) can be used to construct the values of $x_k$, through the use of (\ref{xfuny}). This yields:
\begin{align}
    x_k=y(\frac{\varepsilon k}{V(f)}) \, , \quad k=0,1,\ldots, N(\varepsilon)\, .
\end{align}
The values of $g_k$ follow from (\ref{qzero}) and (\ref{Taylor}):
\begin{align}
     g_k=f((x_{k-1}+x_k)/2) \, , \quad k=1,2,\ldots, N(\varepsilon)\, .
\end{align}
The expressions for $x_k$ and $g_k$ are valid for $\varepsilon \rightarrow 0$. $\qed$

\section*{Appendix B}
The proof of the theorem mentioned in section \ref{exmp} is based on the following
\begin{lemma}[Reverse Hölder inequality]\ \\
Let $p\geq 1$ and $h_1$ and $h_2$ measurable functions on some interval $X$. Define
\begin{align*}
   ||h_1||_{1/p}=(\int_X |h_1(x)|^{1/p}\, \text{d}x)^p \quad \text{and} \quad ||h_2||_{-1/(p-1)}=(\int_X |h_2(x)|^{-1/(p-1)}\, \text{d}x)^{-(p-1)} \, .
\end{align*}
Then,
\begin{align*}
    ||h_1 h_2||_1\geq ||h_1||_{1/p}||h_2||_{-1/(p-1)} \, .
\end{align*}
\end{lemma}
The lemma follows from the standard Hölder inequality.\\
When applied to $h_1(x)=f'(x)$, $h_2(x)=1$, $p=2$ and $X=[c,d]$, the inequality yields:
\begin{align}
\label{revH}
    \int_c^d |f'(x)|\, \text{d}x \geq (\int_c^d|f'(x)|^{1/2}\,\text{d}x)^2(d-c)^{-1} \, .
\end{align}
For monotonically increasing $f(x)$, it follows that $|f'(x)|=f'(x)$ and (\ref{revH}) can be written as
\begin{align}
   (\int_c^d|f'(x)|^{1/2}\,\text{d}x)^2\leq  (f(d)-f(c))(d-c)\, .
\end{align}
Define 
\begin{align*}
    \theta(x)&=\begin{cases} 
    0 & \text{if } -1 \leq x <0 \\
    1/2  & \text{if } x=0\\
   1 & \text{if } 0 < x \leq 1  \, .
        \end{cases}
\end{align*}
\begin{theorem}\ \\
Let $\{f_n(x)\}$, $n=1,2,\ldots$ be a sequence of monotonically increasing differentiable functions on $[-1,1]$, such that $\lim_{n\rightarrow \infty} ||f_n(x)-\theta(x)||_1=0$. Then
 \begin{align*}
     \lim_{n\rightarrow \infty} V(f_n(x))=0 \,.
 \end{align*}
\end{theorem}
{\bf{Proof}}: Let $-1<c<0<d<1$. Write
\begin{align}
\label{triple}
  V(f_n)=(\int_{-1}^c |f'_n(x)|^{1/2}\, \text{d}x+\int_{c}^d |f'_n(x)|^{1/2}\, \text{d}x+\int_{d}^1 |f'_n(x)|^{1/2}\, \text{d}x )^2 \, .
\end{align}
It follows from (\ref{revH}) that 
\begin{align*}
    \int_{-1}^c |f'_n(x)|^{1/2}\, \text{d}x \leq ((f_n(c)-f_n(-1))(c+1))^{1/2}.
\end{align*}
Since $\lim_{n\rightarrow \infty} ||f_n(x)-\theta(x)||_1=0$, it follows that 
\begin{align*}
     \lim_{n\rightarrow \infty} ((f_n(c)-f_n(-1))(c+1))^{1/2}=0\, ,
\end{align*}
for all $-1<c<0$. Hence,
\begin{align*}
     \lim_{n\rightarrow \infty} \int_{-1}^c |f'_n(x)|^{1/2}\, \text{d}x=0\, ,
\end{align*}
for all $-1<c<0$. Similarly,
\begin{align*}
     \lim_{n\rightarrow \infty} \int_{d}^1 |f'_n(x)|^{1/2}\, \text{d}x=0\, ,
\end{align*}
for all $0<d<1$. Finally,
\begin{align}
\label{sndtrm}
    \int_{c}^d |f'_n(x)|^{1/2}\, \text{d}x \leq ((f_n(c)-f_n(-1))(d-c))^{1/2}\leq (d-c)^{1/2}\, ,
\end{align}
For all $-1<c<0<d<1$.\\
Hence, by choosing appropriate $c<0$ and $d>0$, all terms in (\ref{triple}) can be made arbitrarily small. $\qed$ 
\newpage
\printbibliography
\end{document}